\documentclass[journal]{IEEEtran}

\usepackage{graphicx}
\usepackage{mathtools}
\usepackage{amsmath,amssymb,amsfonts,amsthm}

\usepackage{bm,algorithm,algpseudocode,cite}
\allowdisplaybreaks
\usepackage[acronyms,nonumberlist]{glossaries}
\usepackage{array}
\usepackage{cite}
\usepackage{pifont}
\usepackage{nicefrac}
\usepackage{tikz,circuitikz}
\usepackage{tikzscale}
\usepackage{comment}

\newcommand{\tick}{\ding{51}}
\newcommand{\cross}{\ding{55}}

\newtheorem{lemma}{Lemma}

\newacronym{NOMA}{NOMA}{non-orthogonal multiple access}
\newacronym{SIC}{SIC}{successive interference cancellation}
\newacronym{PD}{PD}{power-domain}
\newacronym{CD}{CD}{code-domain}
\newacronym{LoS}{LoS}{line-of-sight}
\newacronym{NLoS}{NLoS}{non-LoS}
\newacronym{SINR}{SINR}{signal-to-interference-plus-noise-ratio}
\newacronym{BS}{BS}{base station}
\newacronym{SE}{SE}{spectral efficiency}
\newacronym{CSI}{CSI}{channel state information}
\newacronym[plural=UEs,
            longplural={user equipment}]{UE}{UE}{user equipment}
\newacronym{OMA}{OMA}{orthogonal multiple access}
\newacronym{PA}{PA}{pinching antenna}
\newacronym{PAS}{PAS}{pinching-antenna system}
\newacronym{PSO}{PSO}{particle swarm optimization}
\newacronym{JGD}{JGD}{joint gradient descent}
\newacronym{JO}{JO}{joint optimization}
\newacronym{AO}{AO}{alternating optimization}
\newacronym{SCA}{SCA}{successive convex approximation}
\newacronym{SOC}{SOC}{second order cone}
\newacronym{SOCP}{SOCP}{second order cone programming}
\newacronym{SDP}{SDP}{semi-definite program}
\newacronym{DL}{DL}{downlink}
\newacronym{UL}{UL}{uplink}
\newacronym{QoS}{QoS}{quality-of-service}
\newacronym{LP}{LP}{linear programming}
\newacronym{EE}{EE}{energy efficiency}
\newacronym{MM}{MM}{majorization-minimization}
\newacronym{KKT}{KKT}{Karush-Kuhn-Tucker}
\newacronym{IPA}{IPA}{interior point algorithm}
\newacronym{MIMO}{MIMO}{multiple-input-multiple-output}
\newacronym{AI}{AI}{artificial intelligence}
\newacronym{FAS}{FAS}{fluid antenna system}
\newacronym{RF}{RF}{radio-frequency}
\newacronym{WG}{WG}{waveguide}
\newacronym{HO}{HO}{heuristic optimization}
\newacronym{UB}{UB}{upper bound}
\newacronym{ULA}{ULA}{uniform linear array}

\begin{document}
\title{Max-Min Rate Fairness Optimization for Multi-User Pinching-Antenna NOMA Systems}
\author{Mahmoud AlaaEldin,~\IEEEmembership{Member,~IEEE,} Amy Inwood,~\IEEEmembership{Member,~IEEE,} Xidong Mu,~\IEEEmembership{Member,~IEEE,} and Michail Matthaiou,~\IEEEmembership{Fellow,~IEEE}



\thanks{The authors are with the Centre for Wireless Innovation (CWI), Queen’s University Belfast, Belfast BT3 9DT, U.K. (e-mail: \{m.alaaeldin, a.inwood, x.mu, m.matthaiou\}@qub.ac.uk.)}
\vspace{-0.3 in}
}
\maketitle

\begin{abstract}
Pinching-antenna systems (PASs) are gaining significant research attention due to their unique ability to overcome signal blockages by repositioning dielectric radiating elements, namely pinching antennas (PAs), along waveguides of the order of meters to create line-of-sight signal paths. As each waveguide is driven by a single radio-frequency (RF) chain, non-orthogonal multiple access (NOMA) is a natural choice for PAS-based multi-user communications. Thus, this paper investigates a PAS-enabled multi-user downlink NOMA system that comprises multiple waveguides, each containing multiple PAs. The PA positions and transmit precoding at the base station are jointly optimized with the objective of maximizing the minimum fairness rate among the users. This optimization problem is highly non-smooth and non-convex due to the rapidly oscillating coherent sums arising from inter-PA interference. To address this issue, a two-stage structured and insightful optimization algorithm is proposed. In the first stage, a coarse optimization of the PA positions and transmit power allocations is carried out using an interior-point algorithm, where the PAs channel phases are neglected, yielding solutions in the neighborhood of the true optima. The second stage is a fine-tuning stage that refines the PA positions and the transmit precoding considering the phase shifts of the PAs channels, and it consists of two sub-stages. A phase-zeroing sub-stage is first carried out to reposition each PA in its neighborhood to align the corresponding channels' phases toward zero, ensuring constructive coherent combining. Then, an alternating sub-stage alternates between forward-backward PA position refinement and successive convex approximation-based complex transmit precoding optimization until convergence to optimize the residual phase shifts of the coherent signals. Simulation results show that the proposed framework significantly outperforms the existing heuristic optimization benchmarks while offering quite lower computational time. The results also demonstrate substantial gains over a comparable multiple-input-multiple-output downlink NOMA system, and provide insights into the impact of the number of PAs, users, and transmit power on system performance.
\end{abstract}

\begin{IEEEkeywords}
Downlink NOMA, max-min fairness rate optimization,  multi-waveguide, pinching-antenna systems.
\end{IEEEkeywords}


\section{Introduction}

With the rapid proliferation of data-intensive technologies such as \gls{AI}, virtual reality, and high-definition video streaming, mobile network traffic is growing at an unprecedented rate. Meeting the capacity demands of next-generation networks increasingly requires operation at higher carrier frequencies. However, higher frequencies introduce significant propagation challenges, including greater attenuation, path loss and sensitivity to blockages \cite{mumtaz}, making a strong \gls{LoS} path between transmitter and receiver essential for reliable communication.

\Glspl{PAS} are an emerging flexible antenna technology designed to establish and maintain robust \gls{LoS} links \cite{ding_flexible_2025}. First proposed by NTT DOCOMO in 2022 \cite{fukuda_pinching_2022}, \glspl{PAS} comprise dielectric particles, known as \glspl{PA}, distributed along dielectric \glspl{WG} that can extend to several meters in length\cite{fukuda_pinching_2022, ding_flexible_2025}. This length is a key differentiator from other flexible antenna paradigms, such as \glspl{FAS} and movable antenna systems, in which antenna displacement is typically constrained to the order of wavelengths \cite{ding_flexible_2025}. The ability to reposition \glspl{PA} over long distances enables \glspl{PAS} to overcome large-scale propagation impairments by forming direct \gls{LoS} links, thereby facilitating reliable operation at higher carrier frequencies. Furthermore, \glspl{PAS} are low-cost and mechanically simple, lending themselves well to practical deployment.

\Gls{NOMA} is another technology that has attracted considerable attention as a means of addressing the growing demand for mobile network capacity. Widely regarded as a key enabler for sixth-generation (6G) networks and beyond \cite{elbayoumi_NOMA_2020, maraqa_survey_2020, islam_power_2016}, \gls{NOMA} improves spectral efficiency by allowing multiple users to share the same time-frequency resources simultaneously, thereby making more efficient use of the available wireless spectrum. In \gls{PD}-\gls{NOMA}, multiple users are allocated different power levels according to their respective channel conditions relative to the \gls{BS} while sharing a common time-frequency resource block. Then, at the receiver side, \gls{SIC} is applied to separate and decode the superimposed signals \cite{yahya_error_2023}.

In \glspl{PAS}, each \gls{WG} is typically fed by a single \gls{RF} chain. Since the number of \glspl{WG} is limited in a \gls{PAS} that typically serves more users, \gls{PD}-\gls{NOMA} can be utilized to simultaneously transmit a higher number of data streams than there are \gls{RF} chains, making \gls{NOMA} and \glspl{PAS} complementary technologies \cite{ding_flexible_2025}. Therefore, designing systems involving both technologies is an important research direction.


\vspace{-0.1 in}

\subsection{Related Work}
Multiple studies have considered the integration of \glspl{PAS} in \gls{NOMA} systems. The majority have considered systems with one \gls{WG}, or a single \gls{PA} per \gls{WG} if multiple are considered, and these are summarized below.

\subsubsection{Single WG scenarios}
The works in \cite{xie_low_2025, zeng_sum_2025, zeng_energy_2025} considered a \gls{PAS}-assisted \gls{NOMA} system with a single \gls{PA} located on a single \gls{WG} serving multiple \glspl{UE}. The authors of \cite{xie_low_2025} proposed a low-complexity \gls{DL} design that optimizes the \gls{PA} position to minimize the path loss, based on a derived closed-form expression. Power allocation was then performed among the \glspl{UE} to ensure their minimum rate requirements, while allocating the remaining power budget to the \gls{UE} with the most favorable channel. In \cite{zeng_sum_2025} and \cite{zeng_energy_2025}, the authors proposed an \gls{AO} method to maximize the sum rate and the \gls{EE} of an \gls{UL} system, respectively, where the power allocation coefficients were iteratively optimized with the \gls{PA} position using \gls{PSO} under minimum \gls{QoS} requirements for the \glspl{UE}.

Moreover, both \cite{zhou_sum_2025} and \cite{xu_qos_2025} investigated a \gls{DL} \gls{PAS}-assisted \gls{NOMA} system, where multiple \glspl{PA} on a single \gls{WG} served only two \glspl{UE}. In \cite{zhou_sum_2025}, \gls{AO} was utilized to jointly optimize the \gls{PA} position and the power allocation coefficients to maximize the sum rate, where a bisection-based search was used to optimize the \gls{PA} positions. However, the authors of \cite{xu_qos_2025} maximized the rate of one \gls{UE} considering the required \gls{QoS} of the other by iteratively optimizing the \gls{PA} position and the power allocation coefficients using \gls{AO} and \gls{SCA}. In \cite{wang_two_2025}, a \gls{PAS}-assisted \gls{DL} \gls{NOMA} system consisting of a single \gls{WG} serving multiple \glspl{UE} was studied, where both single \gls{PA} and multiple \glspl{PA} scenarios were considered. A two-stage approach was proposed where \gls{PSO} was used to optimize the \gls{PA} position in the single \gls{PA} case, while a cluster of \glspl{PA} was placed near each \gls{UE} in the multi-\gls{PA} case, followed by the power allocation among the \glspl{UE}. On the other hand, the works in \cite{wang_antenna_2025, xiao_pass_2026, mohammadzadeh_efficient_2025} considered a \gls{DL} \gls{PAS}-assisted \gls{NOMA} system consisting of a single \gls{WG} with multiple \glspl{PA} serving multiple \glspl{UE}. The authors of \cite{wang_antenna_2025} used matching theory to maximize the sum rate of a system involving a discrete \gls{PAS}, where each \gls{PA} can be activated at a discrete set of positions along the \gls{WG}. Besides, the work in \cite{xiao_pass_2026} investigated the capacity limit of \glspl{PAS} with discrete and continuous antenna positioning. In the discrete case, an exhaustive search was performed to determine the optimal antenna activations, followed by \gls{MM} for power allocation. In the continuous case, the large-scale \gls{PA} positioning and power allocation coefficients were jointly optimized via \gls{AO} and \gls{MM}, followed by a fine-tuning stage. Lastly, in \cite{mohammadzadeh_efficient_2025}, a power allocation framework was presented to minimize the total transmit power subject to the user’s individual \gls{QoS} constraints.


\subsubsection{Multi-WG scenarios}
In \cite{fu_power_2025}, a multi-\gls{WG} \gls{PAS}-aided \gls{UL} \gls{NOMA} system was studied, where each \gls{WG} deployed a single \gls{PA} and served two \glspl{UE}. A distributed iterative algorithm was proposed to minimize the transmit power by jointly optimizing the \gls{PA} positions and \gls{SIC} decoding order. In \cite{wang_antenna_2025_2}, a \gls{PAS}-assisted \gls{NOMA} system involving multiple \glspl{WG} were investigated, where each \gls{WG} was assigned to serve one \gls{UE} and contained multiple \glspl{PA} that could be activated at discrete positions. A game-theoretic algorithm was employed to jointly optimize the \gls{WG} assignment and \gls{PA} activation followed by an \gls{SCA}-based power allocation to maximize the sum rate. Finally, the authors of \cite{gan_joint_2025} studied a \gls{PAS}-enabled \gls{DL} \gls{NOMA} system, where multiple \glspl{WG} containing multiple \glspl{PA} served clusters of multiple \glspl{UE}. They proposed a \gls{JO} framework of the transmit beamforming and the \glspl{PA} placement to minimize the total transmit power. Both a gradient-based method, using \gls{MM} and a penalty dual decomposition algorithm, and a \gls{PSO}-based method were proposed. However, the authors showed that the \gls{PSO}-based approach can significantly outperform the gradient-based method due to the tendency of the latter to converge to local minima.



\begin{table*}[t]
  \caption{contributions of this work compared to existing literature on PAS-NOMA optimization methods}
    \centering
    \begin{tabular}{|c|c|c|c|c|c|c|c|c|c|c|c|c|c|}
    \hline
    
        \textbf{Contributions} & \textbf{This Paper} & \cite{xie_low_2025} & \cite{zeng_sum_2025} & \cite{zeng_energy_2025} & \cite{wang_two_2025} & \cite{zhou_sum_2025} & \cite{xu_qos_2025} & \cite{wang_antenna_2025} & \cite{xiao_pass_2026} & \cite{mohammadzadeh_efficient_2025} & \cite{fu_power_2025} & \cite{wang_antenna_2025_2} & \cite{gan_joint_2025} \\ \hline
        
        
        Multiple \glspl{WG} & \tick & \cross & \cross & \cross & \cross & \cross & \cross & \cross & \cross & \cross & \tick & \tick & \tick \\ \hline
        
        Multiple \glspl{PA} per \gls{WG} & \tick & \cross & \cross & \cross & \tick & \tick & \tick & \tick & \tick & \tick & \cross & \tick & \tick \\ \hline
        
        General number of \glspl{UE} per \gls{WG} & \tick & \tick & \tick & \tick & \tick & \cross & \cross & \tick & \tick & \tick & \cross & \cross & \tick \\ \hline
        
        
        Continuous \gls{PA} placement optimization & \tick & \tick & \tick & \tick & \tick & \tick & \tick & \cross & \tick & \cross & \cross & \cross & \tick \\ \hline
        
        
        Transmit beamforming design & \tick & \cross & \cross & \cross & \cross & \cross & \cross & \cross & \cross & \cross & \cross & \cross & \tick\\ \hline

        Structured mathematical optimization for PAs & \tick & \tick & \cross & \cross & \cross & \tick & \tick & \cross & \tick & \cross & \cross & \tick & \cross \\ \hline
        
        Fairness optimization & \tick & \cross & \cross & \cross & \cross & \cross & \cross & \cross & \cross & \cross & \cross & \cross & \cross \\ \hline
    \end{tabular}
    \vspace{-0.1 in}
    \label{tab:comparison}
\end{table*}


\subsection{Motivations and Contributions} \label{motv}

It can be observed from the literature review that most existing \gls{PAS}-aided \gls{NOMA} studies focused on single-\gls{WG} systems, while only a few considered multi-\gls{WG} deployments. Some multi-\gls{WG} works further simplify the problem by assigning each \gls{WG} to one \gls{UE} or a subset of \glspl{UE}, effectively decomposing the design into several single-\gls{WG} subproblems. However, in practical \gls{PAS}-aided \gls{DL} \gls{NOMA} systems, all \glspl{UE} receive the coherent superposition of the signals radiated by all active \glspl{PA} on all \glspl{WG}, which strongly couples the \gls{PA} positions through the received powers, interference terms, and \gls{NOMA}-\gls{SIC} constraints. Moreover, restricting \glspl{PA} to discrete activation points, as in some works, limits the spatial flexibility and achievable performance of \gls{PAS}. The only closely related multi-\gls{WG} multi-\gls{UE} design in \cite{gan_joint_2025} is based on \gls{PSO}, which effectively operates as a black-box search method and provides limited insights into the structure of the optimal \glspl{PA} placement and transmit precoding design.

The optimization of multi-user multi-\gls{WG} \gls{PAS}-aided \gls{NOMA} systems with multiple \glspl{PA} per \gls{WG} is particularly a challenging problem because the number of continuous \gls{PA} placement variables grows with the total number of \glspl{PA} deployed across all \glspl{WG}. In addition, the received signal at each \gls{UE} depends on the coherent combination of all \gls{PA} contributions, which introduces rapid phase-induced fluctuations in the objective function. These effects make direct optimization highly non-smooth and prone to poor local optima. The challenge becomes even more pronounced when the transmit precoding matrix is jointly optimized with the \gls{PA} positions under the total power and \gls{NOMA}-\gls{SIC} constraints. While \gls{HO} methods such as \gls{PSO} can partially address this difficulty, their performance strongly depends on random initialization, population size, and the number of function evaluations, often requiring repeated independent runs and extensive parameter tuning. This results in high computational burden and limited scalability as the number of \glspl{WG}, \glspl{PA}, and \glspl{UE} increases. These limitations motivate the development of a deterministic, efficient, and physically interpretable optimization framework for the problem.

Motivated by these observations, we propose, to the best of our knowledge, the first max-min rate fairness optimization framework for a multi-user multi-\gls{WG} multi-\gls{PA} \gls{PAS}-aided \gls{DL} \gls{NOMA} system. The proposed formulation allows all \glspl{PA} on all \glspl{WG} to jointly contribute to the received signals of all \glspl{UE} through continuous position optimization. The proposed method exploits the observation that, although coherent inter-\gls{PA} combining creates rapid small-scale fluctuations, the corresponding phase-free channel gains exhibit a smoother large-scale structure governed mainly by path loss and geometry. Accordingly, we develop a structured two-stage algorithm that first performs phase-relaxed coarse optimization to identify a promising region of the search space, and then refines the solution through phase alignment and AO of the \gls{PA} positions and transmit precoding.

Our main contributions are summarized as follows:
\begin{itemize}

\item We formulate a max-min rate fairness optimization problem for a \gls{DL} multi-user multi-\gls{WG} \gls{PAS}-aided \gls{NOMA} system with multiple \glspl{PA} per \gls{WG}. The proposed problem formulation jointly optimizes the \gls{PA} positions and transmit precoding matrix under the total transmit power and \gls{NOMA}-\gls{SIC} constraints. To overcome the rapid fluctuating and non-smooth nature of this problem, we develop a two-stage structured optimization algorithm that exploits the structure of the \gls{PAS} channels.

\item In the first stage, we develop an \gls{IPA} to perform coarse optimization to a phase-relaxed version of the problem, where the channel phases are neglected to capture the smooth large-scale behavior of the effective channel gains and obtain a reliable initialization.

\item In the second stage, we introduce a fine-tuning procedure that explicitly accounts for coherent phase combining. This stage first applies a phase-zeroing procedure to align the \gls{PA} channel contributions, and then alternates between forward-backward element-wise \gls{PA} position refinement and \gls{SCA}-based \gls{SOCP} transmit precoding optimization.

\item Simulation results demonstrate that the max-min rate performance of the proposed optimization framework significantly outperforms the \gls{HO} benchmark while consuming less computational time. Additionally, the performance of the proposed optimization scheme is shown to be very close to the upper bound solution obtained from the first stage, which proves the efficacy of the proposed algorithm. Moreover, it is shown that \gls{PAS}-\gls{NOMA} systems provide far superior performance compared to the conventional \gls{MIMO}-\gls{NOMA} counterpart. The results also quantify the impact of key system parameters, such as the number of \glspl{PA}, the number of \glspl{UE}, and the transmit power.

\end{itemize}

A comparison of the contributions of this work with the state-of-the-art is provided in Table \ref{tab:comparison}.

\subsection{Organization and Notations}
The paper is organized as follows: In Sec. \ref{sysmod}, the system model and the problem formulation are introduced. In Sec. \ref{coarse-opt}, we present the coarse \gls{JO} of the transmit beamforming and the PA locations, and the fine tuning stage is illustrated in Sec. \ref{fine-tune}. The simulation results and discussion are provided in Sec. \ref{Sim}, while the conclusions of this work are presented in Sec. \ref{Conc}.

\textit{Notation:}
Bold lowercase letters are used to define vectors; bold uppercase letters are used to define matrices; $|\cdot|$, $\mathfrak{Re}(\cdot)$, $\mathfrak{Im}(\cdot)$, $\arg(\cdot)$, and $(\cdot)^*$ denote the absolute, real, imaginary, angle, and conjugate of a complex number, respectively.


\section{System Model and Problem formulation}  \label{sysmod}

\subsection{System Model}

Consider the \gls{PAS}-aided \gls{DL} \gls{NOMA} system shown in Fig \ref{fig:sys_model}, where $K$ single-antenna \glspl{UE} are served simultaneously by $M$ \glspl{WG}, where $M < K$, of length $L$ m containing $N$ \glspl{PA}. The \glspl{UE} are distributed within a $D_1\times D_2$ m$^2$ region, and the location of \gls{UE} $k$ is denoted $\mathbf{u}_k = [x_k,y_k,0]^T$. The \gls{BS} is located at the origin, and the \glspl{WG} are aligned with the $x$-axis with a vertical displacement of $d$ m. The \glspl{WG} are assumed to be equally spaced along the $y$-axis at intervals of $d_y$ m, such that the $y$ position of \gls{WG} 1 is $y_1=(M-1)d_y/2$ and that of \gls{WG} $M$ is $y_M=-(M-1)d_y/2$. Therefore, the position of the $n$-th \gls{PA} on the $m$-th \gls{WG} is denoted $\mathbf{z}_m(x_{m,n}) = [x_{m,n}, y_m,d]^T$, where $x_{m,n}\in [0,L]$. Each \gls{WG} is fed by a single \gls{RF} chain.

\begin{figure}[ht]
\centering
\resizebox{0.485\textwidth}{!}{\includegraphics{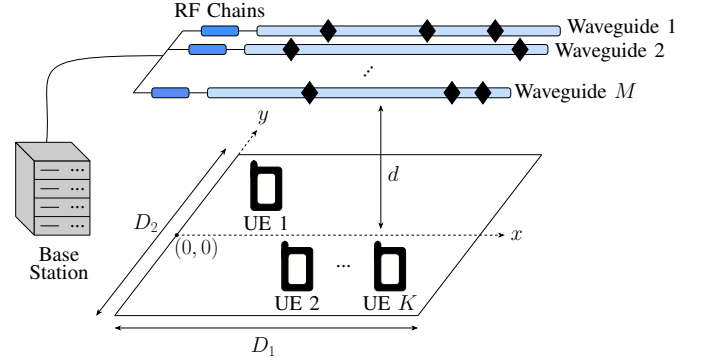}}
\caption{A PAS-based NOMA communication system involving $M$ \glspl{WG} each containing $N$ \glspl{PA} serving $K$ \glspl{UE}.}
\label{fig:sys_model}
\end{figure}

As $L$ is of the order of meters, and \glspl{PAS} are primarily intended for use in high-frequency systems, it is assumed that $x_{m,n}$ can be selected to guarantee a purely \gls{LoS} channel. Accordingly, the channel between \gls{PA} $n$ on \gls{WG} $m$ and \gls{UE} $k$ is modeled by the spherical-wave channel model \cite{zhang_beam_2022, wang_antenna_2025}:
\begin{equation}
\tilde{h}_{k}(x_{m,n}) = \frac{\eta\,\mathrm{e}^{-j\frac{2\pi}{\lambda} \left \| \mathbf{u}_k-\mathbf{z}_{m}(x_{m,n}) \right \|} }{\left \| \mathbf{u}_k - \mathbf{z}_{m}(x_{m,n}) \right \|},
\end{equation}
where $\eta=\frac{c}{4\pi f_c}$, $c$ is the speed of light in a vacuum, $f_c$ is the carrier frequency, and $\lambda$ is the wavelength. Additionally, the loss from the signal traveling within the waveguide can be modeled as \cite{ding_flexible_2025}
\begin{equation}
P_{\mathrm{L}}(x_{m,n}) = 10^{-\frac{\kappa}{20} x_{m,n}} \ \mathrm{e}^{\frac{-2\pi j}{\lambda_\mathrm{g}} x_{m,n} },
\end{equation}
where $\kappa$ is the average attenuation factor along the dielectric waveguide in dB/m, $\lambda_\mathrm{g}=\frac{\lambda}{n_\mathrm{eff}}$ is the guided wavelength, and $n_\mathrm{eff}$ is the effective refractive index of the dielectric waveguide. Hence, the channel between UE $k$ and the BS via the $n$-th \gls{PA} on the $m$-th waveguide is 
\begin{equation}
h_k(x_{m,n}) = P_\mathrm{L}(x_{m,n}) \tilde{h}_k(x_{m,n}). \label{eq:htilde}
\end{equation}
As the signals intended for different \glspl{UE} are transmitted via a single \gls{RF} chain per waveguide, they must be superimposed at the \gls{BS} using different power coefficients prior to transmission. Hence, the received \gls{NOMA} signal at \gls{UE} $k$ is
\begin{equation}
y_k = \sum_{m=1}^M \sum_{n=1}^N {h}_k(x_{m,n}) \sum_{k'=1}^K w_{mk'} s_{k'} + n_k,
\end{equation}
where $s_{k'}$ is the signal intended for \gls{UE} $k'$, $w_{mk'} = a_{mk'} \mathrm{e}^{j\theta_{mk'}}$ is the transmit precoding coefficient applied on \gls{WG} $m$ for the data stream of \gls{UE} $k'$, while $a_{mk'}$ and $\theta_{mk'}$ are the power allocation coefficient and beamforming phase shift for \gls{UE} $k'$ on \gls{WG} $m$, respectively. The precoding coefficients of the \glspl{UE} satisfy the total power constraint at the \gls{BS} as
\begin{equation}  \label{tot_pow}
\sum\nolimits_{m=1}^M \sum\nolimits_{k'=1}^K |w_{mk'}|^2 \leq P_T,
\end{equation}
where $P_T$ is the total transmit power.

In vector form, the received signal can be written as
\begin{equation}
y_k= \mathbf{h}_k^T \mathbf{w}_ks_k + \sum\nolimits_{k' \neq k} \mathbf{h}_k^T \mathbf{w}_{k'}s_{k'} + n_k,
\end{equation}
where $\mathbf{h}_k=[h_{1k}(\mathbf{x}_1), \dots, h_{Mk}(\mathbf{x}_M)]^T$ is the aggregate channel vector from all $M$ waveguides to \gls{UE} $k$, with elements 
\begin{equation}
h_{mk}(\mathbf{x}_m) = \sum_{n=1}^Nh_k(x_{m,n}), \quad \mathbf{x}_m = [x_{m,1}, \dots, x_{m,N}],
\end{equation}
while $\mathbf{w}_k=[w_{1k}, \dots, w_{Mk}]^T$ is the transmit precoding vector for \gls{UE}.

The signal gain corresponding to the message of \gls{UE} $j$ as observed at \gls{UE} $k$ is
\begin{equation}
S_{jk}(\mathbf X,\mathbf W) = \sum\nolimits_{m=1}^M w_{mj} h_{mk}(\mathbf x_m), \label{eq:Sjk_multi}
\end{equation}
where $\mathbf X = [\mathbf x_1^T,\dots,\mathbf x_M^T]^T \in \mathbb{R}_{+}^{M \times N}$ and $\mathbf W = [w_{mk}] \in \mathbb C^{M\times K}$. Without loss of generality, we assume that the \gls{SIC} ordering of the \gls{NOMA} users goes from the weakest user, \gls{UE} $1$, to the strongest user, \gls{UE}. $K$\footnote{Although the \gls{UE}'s effective channel strength depends on both its channel vector $\mathbf{h}_k$ and precoding vector $\mathbf{w}_k$, the desired users' ordering is imposed in the design by inserting the \gls{SIC} constraints in the optimization problem.} Thus, the \gls{SINR} when UE $k \ge j$ decodes the message of UE $j$ is
\begin{equation}
\Gamma_{j \to k}(\mathbf X,\mathbf W) = \tfrac{|S_{jk}(\mathbf X,\mathbf W)|^2}{\sum_{\ell=j+1}^K |S_{\ell k}(\mathbf X,\mathbf W)|^2+\sigma_k^2},
\end{equation}
with the corresponding decoding rate
\begin{equation}
R_{j \to k}(\mathbf X,\mathbf W) = B \log_2\bigl(1 + \Gamma_{j \to k}(\mathbf X,\mathbf W)\bigr),
\end{equation}
where $B$ is the bandwidth in Hz.\footnote{For notational simplicity, we set $B = 1$ Hz throughout the paper, such that the rate is equivalent to the spectral efficiency.}
The own-message rate is $R_k(\mathbf X,\mathbf W)\triangleq R_{k\to k}(\mathbf X,\mathbf W)$. Efficient \gls{SIC} requires satisfying the constraints
\begin{equation}
R_k(\mathbf X,\mathbf W)\le R_{k\to m'}(\mathbf X,\mathbf W),\qquad \forall\,1\le k < m' \le K.
\label{eq:sic_multi}
\end{equation}

\vspace{-0.2 in}
\subsection{Problem Formulation}  \label{prob-form}

The objective of this paper is to optimize the positions of all \glspl{PA} on all \glspl{WG}, $\mathbf{X}$, and the active beamforming matrix, $\mathbf{W}$, so as to maximize the minimum own-user rate under the efficient \gls{SIC} constraints in \eqref{eq:sic_multi} and the total power constraint in \eqref{tot_pow}. Moreover, the \gls{PA} locations on each \gls{WG} must obey the minimum spacing constraints between adjacent \glspl{PA} to mitigate mutual coupling effects, where
\begin{equation}
x_{m,n+1}\!-x_{m,n}\ge q,\,\,\, n=1,\dots,N-1,\,\, m=1,\dots,M,
\end{equation}
with $x_{m,n} \in [0, L], \ \forall m,n$.
The corresponding max-min rate fairness problem can then be written as  \vspace{-0.05 in}
\begin{subequations} \label{prob_form}
\begin{align}
\max_{\mathbf X,\mathbf W} \quad & \min_{k\in\{1,\dots,K\}} R_k(\mathbf X,\mathbf W), \label{eq:orig_multi_obj} \\
\text{s.t.} \quad  & R_{k \to m'}(\mathbf X,\mathbf W)  \geq  R_k(\mathbf X,\mathbf W), \quad \forall 1\!\le\! k\!<m'\!\le\! K, \label{eq:orig_multi_sic} \\
& \left\| \mathbf{W} \right\|_F^2 \le P_T, \label{eq:orig_multi_power} \\
& x_{m,n+1} - x_{m,n}\ge q, \quad \forall m,\; n=1,\dots,N-1, \label{eq:orig_multi_spacing} \\
& 0\le x_{m,n}\le L, \quad \forall m,n. \label{eq:orig_multi_bounds}
\end{align}
\end{subequations}
As the rate expressions depend on the coherent sums across all \glspl{WG} and all \glspl{PA}, $S_{jk}(\mathbf X,\mathbf W)$, which introduces rapid phase-induced fluctuations in the objective function and constraints, problem \eqref{prob_form} is highly non-smooth and non-convex. This makes direct optimization very prone to poor local optimum points. To address these challenges, we are going to develop a two-stage optimization framework which consists of a coarse stage and a fine-tuning stage presented in Sec. \ref{coarse-opt} and Sec. \ref{fine-tune}, respectively.



\section{Stage 1: Coarse optimization stage of transmit precoding and pinching beamforming}  \label{coarse-opt}

In this section, to solve problem \eqref{prob_form}, we first discuss a coarse optimization framework that specializes the general model to amplitude-only per-waveguide precoding (i.e. letting  $\theta_{mk}=0$ in the generic coefficient $w_{mk} = a_{mk}e^{j\theta_{mk}}$), and optimizes real-valued amplitudes $a_{mk} \in \mathbb R$. This transformation removes the rapid fluctuations from the objective function and the constraints expressions, resulting in a smooth optimization problem that can be solved using gradient-based methods.\footnote{The solution of the coarse optimization stage does not reflect the actual achievable fairness rate as it neglects the effect of the phase shifts of the channels of the \glspl{PA} by assuming perfect coherent channel combining at the \glspl{UE}. However, it is an excellent first stage in obtaining the optimized \gls{PA} locations and power allocations for each \gls{WG}, where it is followed by a fine-tuning stage that considers the phase shifts, as discussed in Sec. \ref{fine-tune}.} It also provides the additional benefit of removing the destructive interference possible from the fast phase oscillations of the coherent sums.

\subsection{Joint Optimization Step}

In the coarse optimization stage, the complex channels are replaced by their magnitudes. Define
\begin{equation}
\bar g_{mk}(\mathbf x_m) = \sum\nolimits_{n=1}^N |h_k(x_{m,n})|,
\end{equation}
and the corresponding magnitude-based coherent sum
\begin{equation}
\bar S_{jk}(\mathbf X,\mathbf A) = \sum\nolimits_{m=1}^M a_{mj}\bar g_{mk}(\mathbf x_m). \label{eq:coarse_S_multi}
\end{equation}
The own-user and \gls{SIC} rates in this coarse stage are therefore
\begin{align}
\bar R_k &= \log_2\!\left(1 + \frac{|\bar S_{kk}|^2}{\sum_{\ell=k+1}^K |\bar S_{\ell k}|^2+\sigma^2}\right), \\
\bar R_{k\to m'} &= \log_2\!\left(1 + \frac{|\bar S_{km'}|^2}{\sum_{\ell=k+1}^K |\bar S_{\ell m'}|^2+\sigma^2}\right),  \quad m'>k.
\end{align}

To further obtain a smooth and differentiable objective function, the minimum-rate operator is replaced by the log-sum-exp surrogate
\begin{equation}
\phi_\tau(\mathbf r) = \tau \log \Bigl( \sum\nolimits_{k=1}^K e^{-r_k/\tau}\Bigr), \qquad r_k=R_k(\mathbf X, \mathbf A),
\end{equation}
where $\tau>0$ is small. Since $\phi_\tau(\mathbf r)\to-\min_k r_k$ as $\tau\to 0^+$, minimizing $\phi_\tau$ is equivalent to maximizing a smooth approximation of the minimum user rate.

The \gls{JO} step jointly optimizes the \gls{PA} positions, $\mathbf{X}$, on all \glspl{WG} and the amplitude coefficients, $\mathbf{A}$, for all user streams at the same time without creating two separate sub-problems. Let $\mathbf a {=} \mathrm{vec}(\mathbf A)\in\mathbb R^{MK\times 1}$ and $\mathbf x {=} \mathrm{vec}(\mathbf X)\in\mathbb R^{MN\times 1}$, and define the global decision vector as follows:
\begin{equation}
\mathbf z = [\mathbf x^T, \mathbf a^T]^T \in \mathbb R^{(MN+MK) \times 1}.
\end{equation}
Accordingly, the coarse \gls{JO} over $\mathbf{z}$ can be formulated as
\begin{subequations}
\begin{align}
\min_{\mathbf z}\quad & f_0(\mathbf z) = \tau \log \Bigl( \sum\nolimits_{k=1}^K e^{-\bar R_k(\mathbf z)/\tau}\Bigr),   \label{eq:coarse_multi_obj}\\
\text{s.t.}\quad & c_{k,m'}(\mathbf z)=\bar R_k(\mathbf z)-\bar R_{k\to m'}(\mathbf z)\le 0, \ \forall k<m',  \\
& c_P(\mathbf z) = \sum\nolimits_{m=1}^M \sum\nolimits_{k=1}^K a_{mk}^2 - P_T \le 0,  \\
& \mathbf A_{\mathrm{sp}}\mathbf z\le \mathbf b_{\mathrm{sp}}, \qquad \boldsymbol\ell\le \mathbf z\le \mathbf u,  \label{eq:coarse_multi_lin}
\end{align}
\end{subequations}
where $\mathbf A_{\mathrm{sp}}\mathbf z\le \mathbf b_{\mathrm{sp}}$ stacks the spacing inequalities of all $M$ WGs. Specifically, for each WG $m$, the block associated with $\mathbf x_m$ enforces $x_{m,n}-x_{m,n+1}\le -q$ for $n=1,\dots,N-1$, and the full matrix is obtained by concatenating the $M$ blocks along the diagonal of $\mathbf{A}_{\mathrm{sp}}$.

\vspace{-0.1 in}

\subsection{Interior-Point Solution}
An interior-point strategy solves \eqref{eq:coarse_multi_obj}-\eqref{eq:coarse_multi_lin} by replacing the nonlinear, linear, and bound inequalities by a sequence of logarithmic barrier subproblems. Let the stacked nonlinear constraints be
\begin{equation}
    \mathbf c(\mathbf z)=\bigl[c_{1,2},\dots,c_{K-1,K},c_P\bigr]^T,
\end{equation}
and gather all inequalities into $\tilde{\mathbf c}(\mathbf z)<\mathbf 0$. For barrier parameter $\mu>0$, the penalized objective is
\begin{equation}
\Psi_\mu(\mathbf z)=f_0(\mathbf z) - \mu \sum\nolimits_i \log \bigl( -\tilde c_i(\mathbf z) \bigr).
\label{eq:barrier_multi}
\end{equation}
The logarithmic terms force all iterates to remain strictly feasible while the barrier parameter is gradually reduced.

Introducing the multipliers $\boldsymbol\lambda\succeq\mathbf 0$ for the nonlinear inequalities, $\boldsymbol\nu\succeq\mathbf 0$ for the linear spacing inequalities, and $\mathbf s_L,\mathbf s_U\succeq\mathbf 0$ for the lower and upper bounds, respectively, the Lagrangian is
\begin{multline}
\mathcal L(\mathbf z,\boldsymbol\lambda,\boldsymbol\nu)=f_0(\mathbf z)+\boldsymbol\lambda^T\mathbf c(\mathbf z)+\boldsymbol\nu^T(\mathbf A_{\mathrm{sp}}\mathbf z-\mathbf b_{\mathrm{sp}}) \\
+\mathbf s_L^T(\boldsymbol\ell-\mathbf z)+\mathbf s_U^T(\mathbf z-\mathbf u).
\end{multline}
The perturbed \gls{KKT} conditions are written as
\begin{equation}
\begin{aligned}
&\nabla f_0(\mathbf z)+\mathbf J_c(\mathbf z)^T\boldsymbol\lambda+\mathbf A_{\mathrm{sp}}^T\boldsymbol\nu-\mathbf s_L+\mathbf s_U=\mathbf 0,\\
&\mathbf c(\mathbf z)+\mathbf y=\mathbf 0,\quad \mathbf A_{\mathrm{sp}}\mathbf z-\mathbf b_{\mathrm{sp}}+\mathbf t=\mathbf 0,\\
&\mathbf z-\boldsymbol\ell-\mathbf r_L=\mathbf 0,\quad \mathbf u-\mathbf z-\mathbf r_U=\mathbf 0,\\
&\boldsymbol\Lambda\mathbf y=\mu\mathbf 1,\quad \mathbf N\mathbf t=\mu\mathbf 1,\quad \mathbf S_L\mathbf r_L=\mu\mathbf 1,\quad \mathbf S_U\mathbf r_U=\mu\mathbf 1,
\end{aligned}
\label{eq:kkt_multi}
\end{equation}
where $\mathbf y,\mathbf t,\mathbf r_L$, and $\mathbf r_U$ are strictly positive slack vectors. The first line corresponds to stationarity, the second and third lines impose primal feasibility through the slack representation, and the last line enforces perturbed complementarity. Dual feasibility is captured by the nonnegativity of the multiplier vectors. The Jacobian of the nonlinear constraints is
\begin{equation}
\mathbf J_c(\mathbf z)=\frac{\partial \mathbf c(\mathbf z)}{\partial \mathbf z^T}
=\begin{bmatrix}
\nabla c_1(\mathbf z)^T\\[-0.2em]
\vdots\\[-0.2em]
\nabla c_{K(K-1)/2+1}(\mathbf z)^T
\end{bmatrix}.
\end{equation}
Each row quantifies the first-order sensitivity of one SIC or power constraint to perturbations in all PA positions and all amplitude coefficients across the $M$ WGs.

At iteration $t$, the nonlinear \gls{KKT} residuals are linearized around the current interior point. This means replacing every nonlinear term by its first-order Taylor expansion, which yields a linear system in the increments $(\Delta\mathbf z,\Delta\boldsymbol\lambda,\Delta\boldsymbol\nu,\dots)$. After eliminating the slack increments, the core Newton system takes the saddle-point form
\begin{equation}
\begin{bmatrix}
\mathbf H_L & \mathbf J_c^T & \mathbf A_{\mathrm{sp}}^T \\
\mathbf J_c & -\mathbf W_c & \mathbf 0 \\
\mathbf A_{\mathrm{sp}} & \mathbf 0 & -\mathbf W_A
\end{bmatrix}
\begin{bmatrix}
\Delta\mathbf z \\ \Delta\boldsymbol\lambda \\ \Delta\boldsymbol\nu
\end{bmatrix}
=-
\begin{bmatrix}
\mathbf r_{\mathrm{sta}} \\ \mathbf r_{\mathrm{nl}} \\ \mathbf r_{\mathrm{lin}}
\end{bmatrix},
\label{eq:saddle_multi}
\end{equation}
where $\mathbf H_L\!=\!\nabla^2_{\mathbf z\mathbf z}\mathcal L(\mathbf z,\boldsymbol\lambda,\boldsymbol\nu)$, and $\mathbf W_c$ and $\mathbf W_A$ are positive diagonal matrices produced by the barrier terms. Solving \eqref{eq:saddle_multi} gives the Newton direction. A damped step then preserves strict feasibility and reduces the barrier merit function. Next, the barrier parameter is decreased, and the process repeats.

The gradient of the smooth fairness surrogate is determined by the soft-min weights as follows:
\begin{equation}
\bar\omega_k(\mathbf z) = \frac{e^{-\bar R_k(\mathbf z)/\tau}}{\sum_{i=1}^K e^{-\bar R_i(\mathbf z)/\tau}},  \qquad \sum\nolimits_{k=1}^K\bar\omega_k=1,
\end{equation}
leading to
\begin{equation}
\nabla f_0(\mathbf z) = -\sum\nolimits_{k=1}^K \bar\omega_k(\mathbf z)\nabla \bar R_k(\mathbf z). \label{eq:grad_multi}
\end{equation}
The Hessian matrix of $f_0(\mathbf z)$ is given as
\begin{multline}
\nabla^2 f_0(\mathbf z) = -\sum\nolimits_{k=1}^K \bar\omega_k \nabla^2 \bar R_k + \frac{1}{\tau} \sum\nolimits_{k=1}^K \bar\omega_k \nabla \bar R_k \nabla \bar R_k^T  \\
-\frac{1}{\tau} \Bigl( \sum\nolimits_{k=1}^K \bar\omega_k \nabla \bar R_k \Bigr) \Bigl( \sum\nolimits_{k=1}^K \bar\omega_k \nabla \bar R_k \Bigr)^T,   \label{eq:hess_multi}
\end{multline}
and the Hessian of the Lagrangian is
\begin{equation}
\nabla^2_{\mathbf z\mathbf z}\mathcal L(\mathbf z,\boldsymbol\lambda,\boldsymbol\nu) = \nabla^2 f_0(\mathbf z) + \sum\nolimits_i \lambda_i\nabla^2 c_i(\mathbf z).
    \label{eq:hessL_multi}
\end{equation}
To expose the rate derivatives, define
\begin{align}
Q_k(\mathbf X,\mathbf A) &= |\bar S_{kk}(\mathbf X,\mathbf A)|^2, \\
I_k(\mathbf X,\mathbf A) &= \sum\nolimits_{\ell=k+1}^K |\bar S_{\ell k}(\mathbf X, \mathbf A)|^2 + \sigma^2,
\end{align}
so that $\bar R_k=\log_2(1+Q_k/I_k)$. Letting $\xi_k=Q_k/I_k$, we have
\begin{equation}
    \nabla \bar R_k=\frac{1}{\ln 2}\frac{1}{1+\xi_k}\nabla \xi_k,
\end{equation}
\begin{equation}
    \nabla^2\bar R_k=\frac{1}{\ln 2}\left[\frac{\nabla^2\xi_k}{1+\xi_k}-\frac{\nabla\xi_k\nabla\xi_k^T}{(1+\xi_k)^2}\right],
\end{equation}
with
\begin{equation}
    \nabla\xi_k=\frac{I_k\nabla Q_k-Q_k\nabla I_k}{I_k^2}.
\end{equation}
The derivatives of $Q_k$ and $I_k$ involve the sensitivities of the magnitude-based coherent sums in \eqref{eq:coarse_S_multi} with respect to every $x_{m,n}$ and $a_{mk}$. For example,
\begin{equation}
    \frac{\partial \bar S_{jk}}{\partial a_{mj}}=\bar g_{mk}(\mathbf x_m),
    \qquad
    \frac{\partial \bar S_{jk}}{\partial x_{m,n}}=a_{mj}\frac{\partial \bar g_{mk}(\mathbf x_m)}{\partial x_{m,n}},
\end{equation}
where
\begin{equation}
    \frac{\partial \bar g_{mk}(\mathbf x_m)}{\partial x_{m,n}}=\frac{\partial |h_k(x_{m,n})|}{\partial x_{m,n}}.
\end{equation}
These first- and second-order derivatives are then assembled into \eqref{eq:grad_multi}-\eqref{eq:hessL_multi}, which drive the interior-point Newton iterations. At each iteration, the variable vectors are updated as
\begin{align}
\mathbf z ^{(t+1)} = \mathbf z^{(t)} &+ \alpha_{\mathrm{pri}} \Delta \mathbf z, \quad
\boldsymbol \nu^{(t+1)} = \boldsymbol \nu^{(t)} + \alpha_{\mathrm{dual}} \Delta \boldsymbol \nu, \nonumber  \\
& \boldsymbol \lambda^{(t+1)} = \boldsymbol \lambda^{(t)} + \alpha_{\mathrm{dual}} \Delta \boldsymbol \lambda,
\end{align}
with step sizes, $\alpha_{\mathrm{pri}}$ and $\alpha_{\mathrm{dual}}$, chosen to preserve $\boldsymbol \nu > \boldsymbol 0$ and $\boldsymbol \lambda > \boldsymbol 0$. The barrier parameter $\mu$ is gradually reduced until convergence at a \gls{KKT} point is reached. This primal-dual interior-point framework enables efficient \gls{JO} of the transmit powers, $\mathbf{a}$, and the geometric parameters, $\mathbf{X}$, under the nonlinear \gls{NOMA} \gls{SINR} constraints for any $K$.

\vspace{-0.1 in}
\subsection{Optimal Real-valued Precoding Step}




After performing the \gls{JO} step using the interior point algorithm, we test its solution by deriving the optimal precoding vectors at the \gls{BS} given the optimized \gls{PA} locations of the \gls{JO} step. Specifically, we solve a total sum-power minimization problem such that all the \gls{NOMA} rates, including the \gls{SIC} rates, are greater than the max-min achievable rate obtained from the \gls{JO} step in the previous subsection. Since we assume real and element-wise positive channels in the coarse optimization stage, the optimal precoding vectors should also be real. Fortunately, we prove that the \gls{NOMA} real precoding case has an optimal solution with a specific structure.

Let $\mathbf{X}^{\mathrm{Joint}}$ denote the \gls{PA} locations obtained from the coarse \gls{JO} stage, and let $\gamma^{\mathrm{Joint}}$ denote the corresponding max-min achievable \gls{SINR}. Since the coarse stage is carried out using real and element-wise nonnegative effective channels, the subsequent active-precoding refinement can be restricted to real-valued precoding vectors without loss of optimality. In particular, for \gls{UE} $k$, define the real aggregate channel vector
\begin{equation}
\mathbf{h}_k=
\Big[h_{1k}(\mathbf{x}^{\mathrm{Joint}}_1),\dots,h_{Mk}(\mathbf{x}^{\mathrm{Joint}}_M)\Big]^T
\in \mathbb{R}_+^{M},
\end{equation}
where $\mathbf{x}^{\mathrm{Joint}}_m$ is the optimized \gls{PA}-location vector on \gls{WG} $m$. Moreover, let $\mathbf{a}_k\in\mathbb{R}^{M}$ denote the real precoding vector of \gls{UE} $k$, and collect all user precoders in the real precoding matrix
\begin{equation}
\mathbf{A} \triangleq \big[\mathbf{a}_1,\dots,\mathbf{a}_K\big]\in\mathbb{R}^{M\times K}.
\end{equation}
For a fixed target $\gamma^{\mathrm{Joint}}$, the real-valued precoding refinement solves the sum-power minimization problem
\begin{subequations}
\label{eq:real_precoding_power_min}
\begin{align}
\min_{\mathbf{A}}\quad & \|\mathbf{A}\|_F^2,
\label{eq:real_precoding_power_min_obj}\\
\text{s.t.}\quad
& \frac{|\mathbf{h}_m^T\mathbf{a}_k|^2}
{\sum_{\ell=k+1}^K |\mathbf{h}_m^T\mathbf{a}_\ell|^2+\sigma_m^2}
\!\ge\! \gamma^{\mathrm{Joint}},
\quad \forall\, 1\!\le\! k\!\le\! m\!\le\! K.
\label{eq:real_precoding_power_min_sinr}
\end{align}
\end{subequations} \vspace{-0.05 in}
Problem \eqref{eq:real_precoding_power_min} checks whether the coarse-stage \gls{SINR} target can be obtained with minimum transmit power after optimally re-adjusting the active precoder for the fixed coarse-stage \gls{PA} positions.

\begin{lemma}
\label{lem:real_precoding_socp}
For the real-positive channel model induced by $\mathbf{X}^{\mathrm{Joint}}$, problem \eqref{eq:real_precoding_power_min} admits an optimal real solution and can be equivalently reformulated as a convex \gls{SOCP}. Moreover, if $\boldsymbol{\mu}^\star=\{\mu_{m,k}^\star\}_{m\ge k}$ denotes an optimal dual solution, then the optimal precoding vectors satisfy
\begin{equation}
\mathbf{a}_k^\star \in \mathrm{Null} \big( \mathbf{Q}_k(\boldsymbol{\mu}^\star) \big),\qquad k=1,\dots,K,
\label{eq:ak_nullspace}
\end{equation}
where
\begin{equation}
\mathbf{Q}_k(\boldsymbol{\mu})=
\mathbf{I}
+\sum_{j<k}\sum_{m=j}^K \mu_{m,j}\gamma^{\mathrm{Joint}}\,\mathbf{h}_m\mathbf{h}_m^T
-\sum_{m=k}^K \mu_{m,k}\,\mathbf{h}_m\mathbf{h}_m^T,
\label{eq:Qk_def}
\end{equation}
and $\boldsymbol{\mu}^\star$ is the solution of the dual \gls{SDP} \vspace{-0.1 in}
\begin{subequations}  \label{eq:dual_SDP}
\begin{align}
\max_{\boldsymbol{\mu}}\quad & \sum\nolimits_{k=1}^K\sum\nolimits_{m=k}^K \mu_{m,k}\gamma_k \sigma_m^2, \\
\text{s.t.}\quad & \mu_{m,k}\ge 0,\qquad \forall m\ge k,\\
& \mathbf{Q}_k(\boldsymbol{\mu}) \succeq \mathbf{0},\qquad k = 1,\dots,K.
\end{align}
\end{subequations}
In the generic case where $\mathrm{Null}(\mathbf{Q}_k(\boldsymbol{\mu}^\star))$ is one-dimensional, one can write
\begin{equation}
\mathbf{a}_k^\star=\sqrt{p_k^\star}\,\mathbf{v}_k,
\label{eq:ak_dir_power}
\end{equation}
where $\mathbf{v}_k$ is the normalized eigenvector associated with the smallest eigenvalue of $\mathbf{Q}_k(\boldsymbol{\mu}^\star)$, and the optimal powers $\{p_k^\star\}$ are obtained by the backward recursion
\begin{equation}
p_K^\star=\max_{m\ge K}\frac{\gamma^{\mathrm{Joint}}\sigma_m^2}{g_{m,K}},
\qquad
g_{m,k}\triangleq (\mathbf{h}_m^T\mathbf{v}_k)^2,
\label{eq:pK_star_joint}
\end{equation}
\begin{equation}
p_k^\star \! = \max_{m\ge k} \frac{\gamma^{\mathrm{Joint}} \left(\!\sum_{\ell=k+1}^K p_\ell^\star g_{m,\ell}\! + \!\sigma_m^2\!\right)} {g_{m,k}}, \quad \!k=K\!-\!1,\dots,1. \label{eq:pk_star_joint}
\end{equation}
\end{lemma}

\begin{IEEEproof}
The proof is provided in Appendix A.
\end{IEEEproof}

Lemma~\ref{lem:real_precoding_socp} provides a computationally efficient post-processing step after the coarse \gls{JO}. Specifically, once $\mathbf{X}^{\mathrm{Joint}}$ and $\gamma^{\mathrm{Joint}}$ are obtained, the minimum required real transmit power is computed from \eqref{eq:real_precoding_power_min}. Let $P_{\min}^{\mathrm{Joint}}\triangleq \|\mathbf{A}^\star\|_F^2$ denote the optimized value. If $P_{\min}^{\mathrm{Joint}}=P_T$, then the coarse-stage target already fully exploits the available transmit power, and the real-valued active-precoding verification is complete; the algorithm then proceeds directly to the fine-tuning stage. However, if $P_{\min}^{\mathrm{Joint}}<P_T$, then the available power budget is not fully utilized and the target \gls{SINR} can be further increased.

In this case, a short outer bisection can be applied over a refined target $\gamma\ge \gamma^{\mathrm{Joint}}$. At each bisection iteration, one solves \eqref{eq:real_precoding_power_min} with $\gamma^{\mathrm{Joint}}$ replaced by $\gamma$, and computes the corresponding minimum power $P_{\min}(\gamma)$. If $P_{\min}(\gamma)\le P_T$, then $\gamma$ is feasible and the lower bisection bound is increased; otherwise, the upper bound is decreased. The iterations continue until $P_{\min}(\gamma)$ reaches $P_T$ within a prescribed tolerance. The resulting refined target is then used as the final output of the real-valued precoding refinement before entering the subsequent fine-tuning stage.


\vspace{-0.1 in}
\subsection{Computational Complexity Analysis of the Coarse Stage} \label{subsec:complexity_multiwg}

The proposed optimization consists of the coarse JO of the PA locations and real amplitudes, followed by the real-valued active-precoding refinement for the fixed PA locations. Let
\begin{equation}
 n_z=M(N+K),\quad C_{\rm nl}=\tfrac{K(K-1)}{2}+1,\quad C_{\rm sp}=M(N-1),
\end{equation}
where $n_z$ is the number of coarse JO variables, $C_{\rm nl}$ is the number of nonlinear SIC and power constraints, and $C_{\rm sp}$ is the number of PA-spacing constraints. After eliminating the slack variables, the effective Newton-system dimension is
\begin{equation}
 n_{\rm IP}=n_z+C_{\rm nl}+C_{\rm sp}.
\end{equation}
At each interior-point iteration, evaluating the PA-user channel magnitudes and derivatives costs $\mathcal{O}(MNK)$, forming the coherent sums costs $\mathcal{O}(MK^2)$, and evaluating the own-user and SIC rates costs at most $\mathcal{O}(K^3)$. In a second-order implementation, assembling the nonlinear Hessian terms requires $\mathcal{O}(C_{\rm nl}n_z^2)$ operations, while the dense Newton factorization costs $\mathcal{O}(n_{\rm IP}^3)$. Hence, for $I_{\rm IP}$ Newton iterations,
\begin{equation}
\begin{aligned}
\mathcal{C}_{\rm JO} &= \mathcal{O} \left(I_{\rm IP}\! \left[MNK+MK^2+K^3+C_{\rm nl}n_z^2+n_{\rm IP}^3 \right] \right) \\
& \simeq \mathcal{O} \left(I_{\rm IP}n_{\rm IP}^3 \right).
\end{aligned}
\label{eq:CJO_compact}
\end{equation}

For the real-valued precoding refinement, we exploit the dual-\gls{SDP}/eigenvector characterization instead of a black-box primal \gls{SOCP} implementation. For a target \gls{SINR} $\gamma$, the number of decoding constraints, and hence the number of dual variables $\{\mu_{m,k}\}_{m\ge k}$, is $C_{\rm dec}=\frac{K(K+1)}{2}$.

The dual \gls{SDP} has $C_{\rm dec}$ scalar variables and $K$ semidefinite constraints $\mathbf Q_k(\boldsymbol\mu)\succeq\mathbf0$ of size $M\times M$. Thus, one SDP interior-point iteration costs
\begin{equation}
\mathcal{O}\!\left(C_{\rm dec}^2KM^2+C_{\rm dec}KM^3+C_{\rm dec}^3\right),
\end{equation}
and obtaining $\boldsymbol\mu^\star$ requires
\begin{equation}
\mathcal{C}_{\rm dual}
=\mathcal{O}\!\left(I_{\rm SDP}
\left[C_{\rm dec}^2KM^2+C_{\rm dec}KM^3+C_{\rm dec}^3\right]\right).
\end{equation}
Afterwards, constructing all $\mathbf Q_k(\boldsymbol\mu^\star)$ costs $\mathcal{O}(KC_{\rm dec}M^2)$ if $\{\mathbf h_m\mathbf h_m^T\}$ are precomputed, extracting the $K$ smallest-eigenvalue eigenvectors costs $\mathcal{O}(KM^3)$, and computing $g_{m,k}=(\mathbf h_m^T\mathbf v_k)^2$ followed by the backward power recursion costs $\mathcal{O}(K^2M)$. With $I_{\rm bis} =
\left\lceil \log_2\!\left(\frac{\gamma_{\max}-\gamma_{\min}}{\epsilon_\gamma}\right) \right\rceil$, the real-valued precoding complexity is
\begin{multline}
\mathcal{C}_{\rm RP}
=\mathcal{O}\Big(I_{\rm bis}\big[
I_{\rm SDP}(C_{\rm dec}^2KM^2+C_{\rm dec}KM^3+C_{\rm dec}^3)\\
+KC_{\rm dec}M^2+KM^3+K^2M
\big]\Big).
\label{eq:CRP_compact}
\end{multline}
Compared with a generic primal \gls{SOCP} interior-point step, whose Newton system couples all $MK$ precoding variables and $C_{\rm dec}$ cone constraints, the proposed dual-\gls{SDP} route is more structured: it solves a lower-dimensional dual problem with $C_{\rm dec}$ scalar variables and $K$ small $M\times M$ postive semi-definite blocks, then recovers the primal solution through $K$ eigenvalue decompositions and a low-cost backward recursion. Therefore, for the typical \gls{PAS}-\gls{NOMA} regime with small-to-moderate $K$ and moderate $M$, the dual-\gls{SDP}/eigenvector/recursion implementation has lower practical complexity than a black-box \gls{SOCP} step.


\section{Stage 2: Joint fine tuning of PA placement and transmit precoding optimization} \label{fine-tune}

The solution of the coarse optimization stage represents the performance upper bound of the system, since we ignore the phase shifts of the channels assuming perfect in-phase signals combining. Therefore, this section introduces the second stage of the proposed optimization framework, namely the fine-tuning stage, which refines the solution of the first stage to find a feasible and near-optimal solution. In contrast to the coarse stage where only channel magnitudes are considered, the fine-tuning stage restores the full complex channels and refines the solution through two successive sub-stages: i)~a sequential phase-zeroing procedure over all \glspl{PA} and \glspl{WG}, and ii)~an alternating refinement between forward/backward \gls{PA} placement and complex transmit precoding updates. The fine-tuning stage therefore acts directly on the coherent sums in \eqref{eq:Sjk_multi}, rather than on their magnitude-only approximations.

\subsection{Phase-Zeroing Sub-Stage}
The coarse solution provides the optimized PA locations, $\mathbf{X}^{\mathrm{Joint}}$, that are robust in terms of path loss, but it does not explicitly align the propagation phases. The purpose of phase zeroing is therefore to perturb each PA location toward a local position where the phase of the corresponding complex channel is closer to zero modulo $2\pi$, thereby reducing destructive coherent combination before the alternating refinement begins. Starting from the coarse-stage positions, the $n$-th \gls{PA} on \gls{WG} $m$ is searched over a local interval while all other \glspl{PA} remain fixed. For a candidate location $x$, the phase-zeroing metric is defined as
\begin{equation}
\mathcal{M}_{m,n}(x_{m,n}) = \sum\nolimits_{k=1}^{K}\operatorname{wrap}_{[-\pi,\pi)}^2\!\bigl(\angle h_k(x_{m,n})\bigr),  \label{eq:phase_metric_multi}
\end{equation}
where $h_k(x_{m,n})$ is the complex channel from \gls{PA} $n$ on waveguide $m$ to user $k$ defined in (\ref{eq:htilde}). We then select the candidate $x_{m,n}$ that minimizes \eqref{eq:phase_metric_multi}.

The search interval is directional and keeps a protection margin between adjacent \glspl{PA}. Let $\Delta_x=\lambda/100$ be the search step, $\Delta=t_\lambda\lambda$ the search span, and $\delta$ the spacing margin, where $\delta$ is the minimum guard spacing between adjacent \glspl{PA} required to mitigate coupling effects, while $t_{\lambda}$ represents the wavelength-normalized search interval assigned to each \gls{PA}. To ensure that the subsequent fine-tuning stage retains adequate adjustment freedom, the guard parameter $q$ in \eqref{eq:orig_multi_spacing} must be selected such that $q \geq t_{\lambda}\lambda + \delta$. Thus, in the forward phase-zeroing pass, the admissible candidate set of \gls{PA} $(m,n)$ is  \vspace{-0.05 in}
\begin{equation}  \label{pz_interval}
\mathcal{X}_{m,n}^{\mathrm{pz}}=
\begin{cases}
[x_{m,n}, \, \min(x_{m,n} {+} \Delta, \, x_{m,n+1} {-} \delta)],  & n<N,  \\
[x_{m,N}, \, \min(x_{m,N} {+} \Delta, \, L)], & n=N,
\end{cases}
\end{equation}
which is sampled with step $\Delta_x$. Hence, the update rule is \vspace{-0.1 in}
\begin{equation}
x_{m,n}^{\star}=\arg\min_{x\in\mathcal{X}_{m,n}^{\mathrm{pz}}}\mathcal{M}_{m,n}(x). \label{eq:pz_update}
\end{equation} 
After each local update, the current minimum rate may be re-evaluated using the full complex channels in order to monitor the quality of the phase-zeroed geometry.

The phase-zeroing sub-stage is intentionally simple. It does not solve the original max-min problem directly; rather, it provides a phase-aware initialization for the subsequent alternating optimization. This is important because the coherent sums in \eqref{eq:Sjk_multi} are highly oscillatory with respect to small changes in $x_{m,n}$, and a direct high-dimensional optimization over the full complex model is prone to returning poor local minima.


\subsection{Alternating Forward/Backward PA Placement and Transmit Precoding}
Once the phase-zeroed geometry is obtained, the algorithm alternates between two coupled blocks: i)~sequential \gls{PA}-position refinement under fixed precoding, and ii)~transmit precoding optimization under fixed \gls{PA} positions. The first block is implemented by forward and backward coordinate-search sweeps over all \glspl{WG}, while the second block is solved through a bisection-assisted \gls{SCA} procedure.

\subsubsection{Forward PA Placement}
In the forward sweep, the algorithm scans each \gls{WG} from the first \gls{PA} to the last \gls{PA}. For fixed precoding matrix $\mathbf{W}$ and fixed locations of all other \glspl{PA}, the position of \gls{PA} $(m,n)$ is updated by maximizing the  worst-user rate over its admissible forward interval, namely
\begin{equation}
x_{m,n}^{\star} = \arg\max_{x \in \mathcal{X}_{m,n}^{\mathrm f}} \; \min_{k \in \{1, \dots, K \}} R_k\bigl(x; \, \mathbf{X}_{-(m,n)}, \mathbf{W} \bigr),  \label{eq:forward_update_multi}
\end{equation}
where $\mathbf{X}_{-(m,n)}$ denotes all \gls{PA} locations except $x_{m,n}$, while $\mathcal{X}_{m,n}^{\mathrm f}$ is the feasible search set of the forward \gls{PA} placement step and it is defined exactly as the expression in \eqref{pz_interval}.

For each candidate point, the complex rates induced by the current coherent sums are computed and the point that maximizes the minimum own-user rate is selected. Since only one coordinate is adjusted at a time, the forward sweep is a deterministic coordinate-search procedure over discretized local intervals.

\subsubsection{Backward PA Placement}
The backward sweep follows the same principle, but scans each \gls{WG} from the last \gls{PA} to the first \gls{PA}. For fixed $\mathbf{W}$ and all other \gls{PA} positions, PA $(m,n)$ is updated as
\begin{equation}
x_{m,n}^{\star} = \arg\max_{x \in \mathcal{X}_{m,n}^{\mathrm b}} \; \min_{k \in \{1, \dots, K\}} R_k \bigl(x; \, \mathbf{X}_{-(m,n)},\mathbf{W} \bigr),
    \label{eq:backward_update_multi}
\end{equation}
with feasible interval
\begin{equation}
\mathcal{X}_{m,n}^{\mathrm b} =
\begin{cases}
[\max(x_{m,n} {-} \Delta, \, x_{m,n-1} {+} \delta), \, x_{m,n}],  & n>1,  \\
[\max(x_{m,1} {-} \Delta, \, 0), \, x_{m,1}],  & n=1.
\end{cases}
\end{equation}
The use of both forward and backward sweeps is crucial. Since the feasible interval of each \gls{PA} depends on already updated neighboring \glspl{PA}, a one-directional sweep may introduce a strong directional bias and lock the geometry prematurely. The reverse sweep compensates for this effect by reopening local search opportunities from the opposite side, which empirically improves the final max-min fairness rate.

In both forward and backward sweeps, the objective is evaluated using the full complex rates. Thus, the algorithm refines the geometry directly with respect to coherent beam combination and \gls{SIC} performance, rather than merely with respect to channel magnitudes.


\subsubsection{Transmit Precoding Update via Bisection and SCA}
After a number of forward/backward placement sweeps, the code updates the transmit precoding matrix under fixed PA locations. At the current geometry $\mathbf{X}$, define the aggregate channel matrix
\begin{equation}
    \mathbf{H}(\mathbf{X})=[\mathbf{h}_1,\dots,\mathbf{h}_K]\in\mathbb{C}^{M\times K},
\end{equation}
where the $k$-th column is the WG-wise aggregate channel vector
\begin{equation}
    \mathbf{h}_k=[h_{1k}(\mathbf{x}_1),\dots,h_{Mk}(\mathbf{x}_M)]^T.
\end{equation}
For a target common \gls{SINR} level $\gamma$, the transmit-precoding subproblem seeks the minimum power beamforming matrix that satisfies all NOMA decoding constraints as
\begin{subequations}
\begin{align}
\min_{\mathbf{W}}\quad & \|\mathbf{W}\|_F^2   \label{eq:sca_power_min}  \\
\text{s.t.} \quad & |\mathbf{h}_k^T \mathbf{w}_j|^2 \ge \gamma \Bigl( \sum\nolimits_{\ell=j+1}^{K} |\mathbf{h}_k^T\mathbf{w}_\ell|^2 + \sigma_k^2\Bigr), \nonumber  \\
& \forall j = 1, \dots, K, k = j, \dots, K,  \label{eq:sca_sinr_con}
\end{align}
\end{subequations}
where $\mathbf{w}_j$ is the $j$-th column of $\mathbf{W}$. Constraint \eqref{eq:sca_sinr_con} is the fixed-threshold version of the SIC inequalities and must hold for each message $j$ at each user $k\ge j$.

Problem \eqref{eq:sca_power_min}-\eqref{eq:sca_sinr_con} is still non-convex because the desired-signal terms $|\mathbf{h}_k^T\mathbf{w}_j|^2$ are convex quadratic functions appearing on the left-hand side of a superlevel constraint. The code handles this using SCA. Suppose that at SCA iteration $r$ the current point is $\mathbf{W}^{(r)}$. Let
\begin{equation}
    c_{kj}^{(r)}=\mathbf{h}_k^T\mathbf{w}_j^{(r)}.
\end{equation}
Then, the first-order affine lower bound of $|\mathbf{h}_k^T\mathbf{w}_j|^2$ around $\mathbf{w}_j^{(r)}$ is
\begin{equation}
    |\mathbf{h}_k^T\mathbf{w}_j|^2\ge 2\Re\!\left\{\bigl(c_{kj}^{(r)}\bigr)^*\mathbf{h}_k^T\mathbf{w}_j\right\}-|c_{kj}^{(r)}|^2.
    \label{eq:affine_lb}
\end{equation}
Replacing each desired-signal term in \eqref{eq:sca_sinr_con} by \eqref{eq:affine_lb} yields the convex \gls{SCA} subproblem
\begin{subequations}
\begin{align}
\min_{\mathbf{W}}\quad & \|\mathbf{W}\|_F^2 \\
\text{s.t.}\quad & 2\Re\!\left\{\bigl(c_{kj}^{(r)}\bigr)^*\mathbf{h}_k^T\mathbf{w}_j\right\}-|c_{kj}^{(r)}|^2  \nonumber \\
& \ \ \ge \gamma \! \left(\! \sum\nolimits_{\ell=j+1}^{K} |\mathbf{h}_k^T \mathbf{w}_\ell|^2 + \sigma_k^2 \!\right) \!, \quad \forall j,\; k\ge j.   \label{eq:sca_convex_subproblem}
\end{align}
\end{subequations}
This subproblem is convex because the left-hand side is affine in $\mathbf{W}$, whereas the interference terms on the right-hand side remain convex quadratic expressions.

The SCA procedure requires an initial feasible point. The code constructs it by fixing the normalized beam directions \vspace{-0.05 in}
\begin{equation}
\mathbf{u}_j = \frac{\mathbf{h}_j}{\|\mathbf{h}_j\|},\qquad j=1,\dots,K,
\end{equation}
and searching only over the scalar powers $p_j\ge 0$, so that $\mathbf{w}_j=\sqrt{p_j}\,\mathbf{u}_j$. Substituting this into the fixed-\gls{SINR} inequalities leads to a simple \gls{LP} power-allocation problem as \vspace{-0.05 in}
\begin{subequations}
\begin{align}
\min_{\{p_j\ge 0\}}\quad & \sum\nolimits_{j=1}^{K}p_j  \\
\text{s.t.} \quad \ & A_{kj}p_j \!\ge\! \gamma\! \left( \!\sum\nolimits_{\ell=j+1}^{K}A_{k\ell}p_\ell + \sigma_k^2 \! \right) \!\!,
\quad \forall j,\; k\ge j,
\end{align}  
\end{subequations}
where $A_{kj} = |\mathbf{h}_k^T\mathbf{u}_j|^2$. If this initialization is feasible, the \gls{SCA} iterations proceed until the relative change of the required power becomes smaller than a prescribed tolerance.

\subsubsection{Outer Bisection on the Common SINR Target}
The SCA block solves the minimum-power problem for a \emph{fixed} common SINR target $\gamma$. To maximize the minimum user rate,  an outer bisection search over $\gamma$ is performed. Let $\gamma_{\min}$ and $\gamma_{\max}$ denote the lower and upper bounds on the feasible common SINR. For each trial value $\gamma^{\mathrm{try}}=\frac{\gamma_{\min}+\gamma_{\max}}{2}$, the inner \gls{SCA} problem returns the minimum required power $P_{\mathrm{req}}(\gamma^{\mathrm{try}})$. If $P_{\mathrm{req}}\le P_T$, then the target \gls{SINR} is feasible and the lower bisection bound is increased; otherwise, the upper bound is decreased. The procedure terminates when the \gls{SINR} interval becomes sufficiently small or when the required power approaches the available power budget. The final common minimum rate is $R_{\mathrm{opt}} = \log_2( 1 + \gamma_{\mathrm{opt}})$.

Therefore, the transmit-precoding update consists of three nested layers: i)~outer bisection over the common \gls{SINR} target, ii)~inner power minimization for a fixed target, and iii)~ \gls{SCA} iterations for convexifying the non-convex desired-signal terms.

\begin{algorithm}[t]
\footnotesize
\caption{Fine-Tuning Stage for Multi-Waveguide PAS-NOMA}
\label{alg:fine_tuning_multi}
\begin{algorithmic}[1]
\Require Coarse-stage PA positions $\mathbf{X}^{\mathrm{Joint}}$ and precoder $\mathbf{W}^{\mathrm c}$; search parameters $t_\lambda$, $\delta$, $\Delta_x=\lambda/100$; tolerances for bisection and SCA
\Ensure Refined PA positions $\mathbf{X}^{\star}$ and precoder $\mathbf{W}^{\star}$
\State Initialize $\mathbf{X}\leftarrow \mathbf{X}^{\mathrm{Joint}}$ and $\mathbf{W}\leftarrow \mathbf{W}^{\mathrm c}$
\For{each WG $m=1,\dots,M$}
    \For{each PA $n=1,\dots,N$}
        \State Build candidate set $\mathcal{X}_{mn}^{\mathrm{pz}}$ and update $x_{m,n}$ with \eqref{eq:pz_update};
    \EndFor
\EndFor
\Repeat
    \For{a prescribed number of local geometry sweeps}
        \For{each WG $m$, PA $n$ in forward order}
            \State Update $x_{m,n}$ by solving \eqref{eq:forward_update_multi};
        \EndFor
        
        \For{each WG $m$, PA $n$ in backward order}
            \State Update $x_{m,n}$ by solving \eqref{eq:backward_update_multi};
        \EndFor
    \EndFor
    
    \State Form the aggregate channel matrix $\mathbf{H}(\mathbf{X})$;
    \State Initialize feasible beam directions/powers;
    \State Set bisection bounds on $\gamma$;
    
    \Repeat
        \State Set trial SINR target $\gamma^{\mathrm{try}}$;
        \State Solve the fixed-$\gamma^{\mathrm{try}}$ power minimization by SCA;
        
        \Repeat
            \State Linearize desired-signal terms using \eqref{eq:affine_lb};
            \State Solve the resulting convex subproblem \eqref{eq:sca_convex_subproblem};
            \State Update the linearization point;
        \Until{required power converges/ SCA limit reached}
        
        \State Update the bisection interval according to whether
        \State the required power is below $P_T$ or above;
    \Until{the bisection interval is sufficiently small}
    
    \State Set $\mathbf{W}$ to the final precoder returned by bisection-
    \State SCA;
    
\Until{the minimum own-user rate converges or the maximum number of outer iterations is reached}
\State \Return $\mathbf{X}^{\star}=\mathbf{X}$ and $\mathbf{W}^{\star}=\mathbf{W}$
\end{algorithmic}
\end{algorithm}


\vspace{-0.1 in}
\subsection{Computational Complexity of the Fine-tuning Stage}
\label{subsec:complexity_finetuning_multiwg}

We now characterize the computational complexity of the proposed fine-tuning stage. Let $Q_x = \left\lceil \frac{t_\lambda\lambda}{\Delta_x} \right\rceil + 1$ be the maximum number of grid points examined in each local \gls{PA} search interval, where $\Delta_x$ is the spatial sampling step.

First, in the phase-zeroing sub-stage, each of the $MN$ \glspl{PA} is updated once by searching over at most $Q_x$ candidate locations. For each candidate point, the phase-zeroing metric in \eqref{eq:phase_metric_multi} requires evaluating the complex channel from the considered \gls{PA} to all $K$ users and computing the corresponding wrapped phase errors. Hence, the complexity of the phase-zeroing sub-stage is $\mathcal{C}_{\rm PZ} = \mathcal{O}\!\left( MNQ_xK \right)$.

Next, consider the forward/backward \gls{PA}-placement refinement. In each forward or backward sweep, all $MN$ \glspl{PA} are sequentially visited. For each \gls{PA}, at most $Q_x$ candidate locations are tested, and the candidate that maximizes the current worst-user rate is selected. For a given candidate point, the update of the affected PA-to-user channels costs $\mathcal{O}(K)$. Then, updating the coherent terms $\{S_{jk}\}$ in \eqref{eq:Sjk_multi} under the fixed precoding matrix costs $\mathcal{O}(K^2)$, while evaluating all own-user and \gls{SIC} rates costs at most $\mathcal{O}(K^3)$. Therefore, the complexity of one complete forward/backward pair of \gls{PA}-placement sweeps is
\begin{equation}
\mathcal{C}_{\rm FB} = \mathcal{O}\!\left(2MNQ_x\left(K+K^2+K^3\right)\right) \simeq \mathcal{O}\!\left(MNQ_xK^3\right).  \label{eq:CFB_multi}
\end{equation}
If $I_{\rm sw}$ forward/backward sweep pairs are performed before each precoding update, the total \gls{PA}-placement refinement cost per fine-tuning outer iteration becomes $\mathcal{C}_{\rm PA} = \mathcal{O} \! \left( I_{\rm sw} MN Q_x K^3 \right)$. This expression corresponds to an incremental implementation in which only the affected coherent sums are updated after moving one \gls{PA}.

The transmit-precoding update is performed for fixed \gls{PA} locations through an outer bisection search over the common \gls{SINR} target and an inner \gls{SCA}-based power minimization. The number of bisection iterations is $I_{\rm bis}^{\rm ft} = \left\lceil \log_2\!\left( \frac{\gamma_{\max}-\gamma_{\min}}{\epsilon_\gamma} \right) \right\rceil$, where $\epsilon_\gamma$ is the bisection accuracy. For each trial value of $\gamma$, a feasible initialization is obtained by solving an LP over the $K$ scalar powers subject to $C_{\rm dec}$ linear inequalities. Using a generic interior-point implementation, this initialization costs $\mathcal{C}_{\rm LP} = \mathcal{O}\!\left(I_{\rm LP}(K + C_{\rm dec})^3\right)$, where $I_{\rm LP}$ is the number of \gls{LP} interior-point iterations and $C_{\rm dec}$ is the number of \gls{NOMA} decoding constraints. This term is typically small compared with the subsequent complex precoding \gls{SCA} step.

For the \gls{SCA} step, the optimization variable is the complex matrix $\mathbf{W} \in \mathbb{C}^{M\times K}$, which is equivalent to $2MK$ real scalar variables. The convexified subproblem in \eqref{eq:sca_convex_subproblem} contains $C_{\rm dec}$ convex quadratic constraints. Let $n_{\rm SCA}=2MK+C_{\rm dec}$ denote the effective dense Newton-system dimension of the convex \gls{SCA} subproblem after introducing the required conic/slack variables. At each \gls{SCA} iteration, forming the affine lower bounds in \eqref{eq:affine_lb} and the interference quadratic terms costs at most $\mathcal{O}(C_{\rm dec}KM)$, while solving the resulting dense convex subproblem by an interior-point method costs $\mathcal{O}(n_{\rm SCA}^3)$. Hence, if $I_{\rm SCA}$ \gls{SCA} iterations are used for each bisection trial, the complexity of the precoding update is
\begin{equation}
\begin{aligned}
\mathcal{C}_{\rm W} & = \mathcal{O} \! \left( I_{\rm bis}^{\rm ft} \left[ \mathcal{C}_{\rm LP} + I_{\rm SCA} \left( C_{\rm dec}KM+n_{\rm SCA}^3 \right) \right] \right)  \\
& \simeq \mathcal{O} \! \left( I_{\rm bis}^{\rm ft} I_{\rm SCA} n_{\rm SCA}^3 \right).
\end{aligned}
\label{eq:CW_ft}
\end{equation}
The approximation in \eqref{eq:CW_ft} follows because the dense Newton factorization normally dominates the cost of forming the linearized constraints.

Finally, the total complexity of all the steps of the fine-tuning stage scales as
\begin{equation}  \label{eq:CFT_multi}
\mathcal{C}_{\rm ft} \simeq \mathcal{O}\!\left( I_{\rm ft} \left[ I_{\rm sw}MNQ_xK^3 + I_{\rm bis}^{\rm ft}I_{\rm SCA}n_{\rm SCA}^3 \right] \right),
\end{equation}
where $I_{\rm ft}$ denotes the number of outer fine-tuning iterations. Equation \eqref{eq:CFT_multi} shows that the fine-tuning cost is composed of two dominant terms. The first term scales linearly with the number of \glspl{PA}, $MN$, and corresponds to the local coordinate-search refinement of the \gls{PA} positions. The second term corresponds to the complex active-precoding update and is mainly governed by the dense interior-point solution of the convexified \gls{SCA} subproblems.


\vspace{-0.05 in}

\section{Simulation Results} \label{Sim}
In this section, we verify the proposed optimization scheme through simulations, compare its performance to a range of benchmarks, and investigate the impact of a range of system parameters. For all results shown, we assume that between two and four \glspl{UE} are served. Where not otherwise specified, the \gls{UE} locations are assumed to be $\mathbf{u}_1=(3,-1,0)$, $\mathbf{u}_2=(10,2,0)$, $\mathbf{u}_3=(18,3,0)$ and $\mathbf{u}_4=(24,4,0)$. All \glspl{UE} are served by a \gls{PAS} system involving $M=2$ waveguides at a height of $d=3$ m. Where not otherwise specified, each wavelength is assumed to be of length $L=30$ m. The attenuation factor and effective refractive index of the dielectric are $\kappa=0.1$ dB/m and $\eta_\mathrm{eff}=1.4$, respectively \cite{ding_flexible_2025}. The carrier frequency is $f_c=28$ GHz, and the noise variance is $\sigma^2=-80$ dBm.

The performance of the proposed scheme is evaluated against three benchmarks:
\begin{itemize}

\item \textbf{Perfect combining \gls{UB}}: The \gls{UB} benchmark is the result obtained after the first stage, i.e., the coarse optimization stage. It assumes perfect phase alignment across all channels so that the aggregate channel reduces to a sum of magnitudes rather than a sum of complex coefficients.

\item \textbf{\gls{HO}}: The only other known work to consider a \gls{PAS}-\gls{NOMA} system involving multiple \glspl{PA} across multiple \glspl{WG} employs \gls{PSO}, that is a \gls{HO} approach \cite{gan_joint_2025}. Among the many \gls{HO} methods we evaluated, a pattern-search algorithm yielded the best performance for the parameters used in these results and is therefore used as a benchmark. For the pattern search, $20$ random initial values are generated, and the highest final result is selected.

\item \textbf{Hybrid beamforming \gls{MIMO}-\gls{NOMA}}: We compare the performance of the proposed scheme against the hybrid analog-digital \gls{MIMO}-\gls{NOMA} system to quantify the gains provided by \glspl{PAS}. For a fair comparison, we assume the \gls{MIMO}-\gls{NOMA} system consists of a \gls{ULA} of $MN$ antennas equally spaced at intervals of $\lambda/2$ and $M=2$ \gls{RF} chains each is only connected to a subset of $N$ antennas of the \gls{ULA}. Hybrid beamforming is used, with analog beamforming over the $N$ antennas connected to each \gls{RF} chain followed by digital beamforming across the two \gls{RF} chains.

\end{itemize}
Throughout the whole simulations, the number of multi-start initializations for our proposed \gls{IPA} for the coarse optimization is $4$, while number of random initializations for the \gls{HO} benchmark is $20$.

\vspace{-0.1 in}
\subsection{Convergence Behavior}

\begin{figure*}
\centering
\includegraphics[width=1.8\columnwidth]{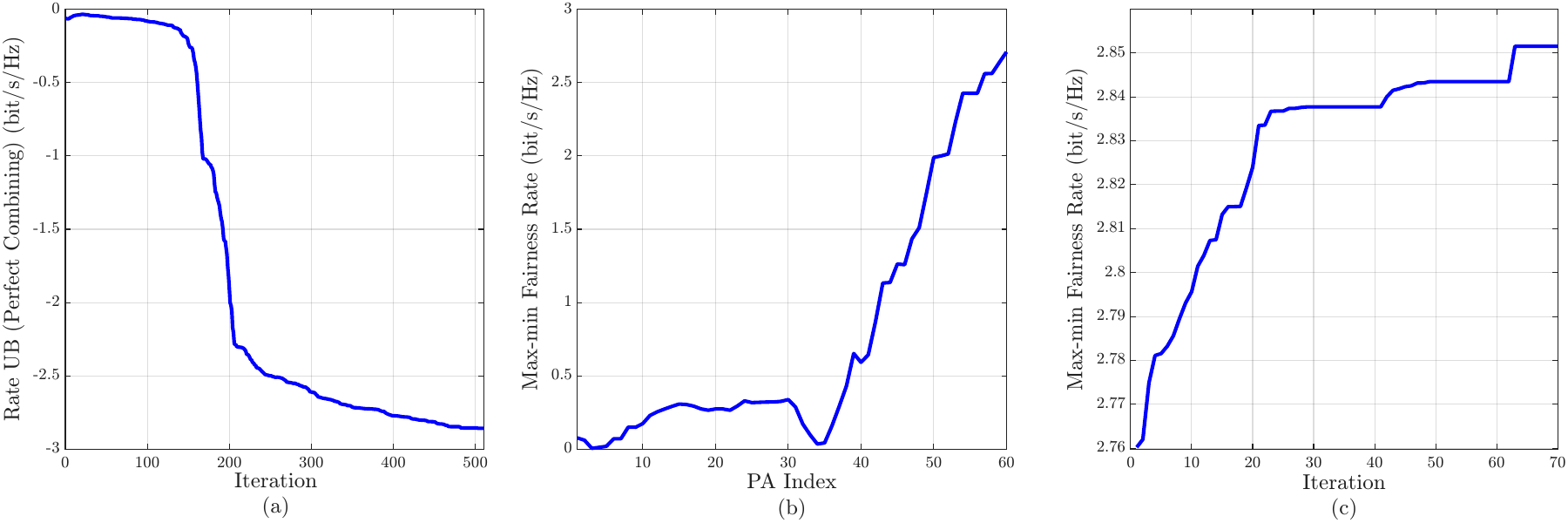}
\vspace{-0.1 in}
\caption{Example convergence behavior for each stage of the proposed optimization scheme: (a) Coarse optimization stage, (b) phase-zeroing sub-stage, (c) alternating forward/backward PA placement and transmit precoding sub-stage.}
\vspace{-0.2 in}
\label{fig:conv}
\end{figure*}

Figure \ref{fig:conv} provides an example of the convergence behavior for each stage of the proposed \gls{JO} algorithm. Figure \ref{fig:conv}(a) shows the perfect-combining \gls{UB} returned at each iteration of the coarse optimization stage, and Fig. \ref{fig:conv}(b) shows the max-min fairness rate as each \gls{PA} is individually repositioned during the phase-zeroing stage. Fig. \ref{fig:conv}(c) shows the max-min fairness rate during the fine-tuning stage, where forward-backward \gls{PA} placement and SCA-based precoder optimization are performed alternately until convergence.

From Fig. \ref{fig:conv}(a), it can be seen that the perfect combining \gls{UB} cost function initially decreases slowly. The interior point method applied takes cautious initial steps while the barrier parameter $\mu$ is still large, ensuring that iterates remain away from the constraint boundaries. Once $\mu$ is reduced sufficiently, the steps become larger and the value drops rapidly towards the optimal. Once $\mu$ is reduced sufficiently, the cost function begins to plateau again as convergence to the \gls{KKT} point occurs. Convergence here is deemed to have occurred when the norm of the gradient function is less than $10^{-7}$.

Figure \ref{fig:conv}(b) illustrates the transition to the true max-min rate objective. Although the \glspl{PA} are initialized at the optimal coarse-stage positions, the max-min rate is initially very low, as the neglected channel phases may combine destructively. The phase-zeroing stage then sequentially repositions each \gls{PA} to align its channel phase toward zero. Performance improves slowly at first, while the majority of \glspl{PA} remain misaligned, before rising rapidly once sufficient phase alignment is achieved across the array.

Figure \ref{fig:conv}(c) shows that the fine-tuning stage begins from the solution provided by the phase-zeroing stage. Forward-backward \gls{PA} position adjustment is first performed until the improvement falls below a prescribed threshold, at which point the \gls{SCA}-based precoder optimization is applied. The updated precoder then creates new opportunities for further positional refinement, and the two steps alternate until convergence. This interleaving produces the step-like behavior observed in Fig. \ref{fig:conv}(c), with each step corresponding to a precoder update followed by a subsequent round of position refinement. Convergence is declared when forward-backward adjustment yields no further improvement even after precoder optimization.

\vspace{-0.1 in}
\subsection{The Impact of the Number of Pinching Antennas}

Figure \ref{fig:N} investigates the impact of the number of \glspl{PA} on the max-min fairness rate obtained by \glspl{UE}, and compares the performance of the proposed optimization approach with the detailed benchmarks. Here, $P_T=3$ dBm. 

\begin{figure}[t]
\centering
\includegraphics[width=\columnwidth]{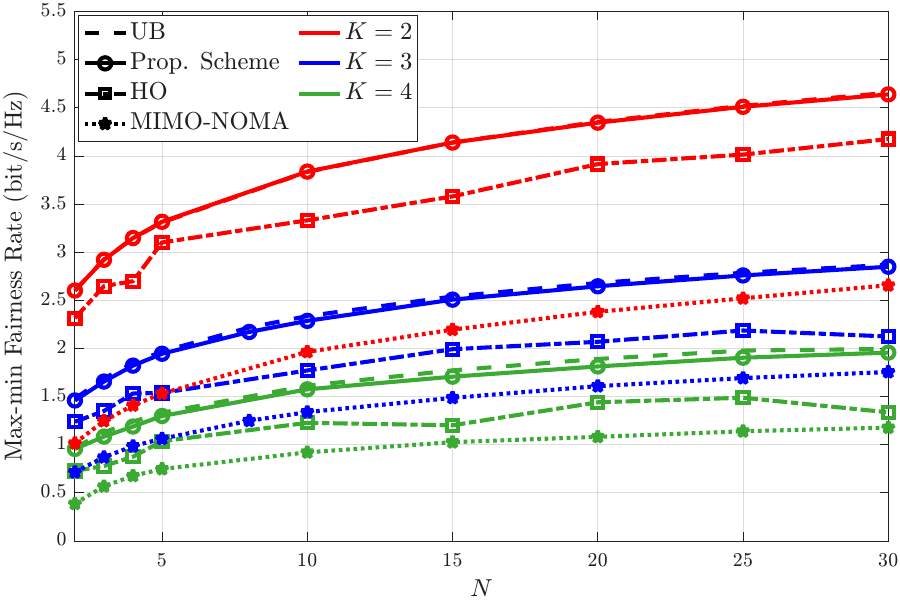}
\vspace{-0.2 in}
\caption{Performance comparison of the proposed scheme against existing benchmarks versus the number of \glspl{PA}.}
\label{fig:N}
\end{figure}

The proposed \gls{JO} method achieves max-min rate performance nearly indistinguishable from the perfect combining \gls{UB}, particularly for low $K$. This is attributed to the reduced number of phase alignment constraints, as serving fewer \glspl{UE} requires the phases of fewer channels to be simultaneously aligned, enabling near-perfect coherent combining. Additionally, it can be seen that the max-min rate curve for the proposed \gls{JO} approach grows smoothly, in contrast to the irregular growth observed for the \gls{HO} method. The guided nature of the proposed approach significantly reduces its reliance on initial position. The \gls{HO} method is far more reliant on the initial position, resulting in frequent convergence to local optima and the irregular behavior seen in Fig. \ref{fig:N}. All the three considered \gls{PAS} schemes significantly outperform the \gls{MIMO} benchmark, even under the favorable assumption of perfect \gls{LoS} from the origin-located \gls{ULA} to all \glspl{UE}. This gain is attributable to the ability of \glspl{PA} to be placed in close proximity to each \gls{UE}, yielding substantial path loss reductions that a fixed-array MIMO system cannot replicate.

As $N$ increases, \glspl{UE} avail of the additional degrees of freedom that increased numbers of radiating elements provide. For \glspl{PAS}, the additional \glspl{PA} can be positioned near to \glspl{UE}, reducing path losses and increasing the aggregate channel magnitude. For MIMO systems, an increase in the number of antennas increases the beamforming gain. In both scenarios, this drives an approximately logarithmic increase in max-min rate. When $K$ is higher, the total transmit power is shared between more \glspl{UE}, and less \glspl{PA} can be located in close proximity to each \gls{UE}, resulting in a drop in the max-min rate. Additionally, the max-min rate increases more slowly with $N$ for higher $K$. In this case, higher inter-user interference is a limiting factor and reduces the \gls{SINR}. The effectiveness of \gls{SIC} is also reduced. Together, these factors reduce the marginal improvements offered by additional antenna elements.

\vspace{-0.05 in}
\subsection{The Impact of the Transmit Power}

\begin{figure}
\centering
\includegraphics[width=\columnwidth]{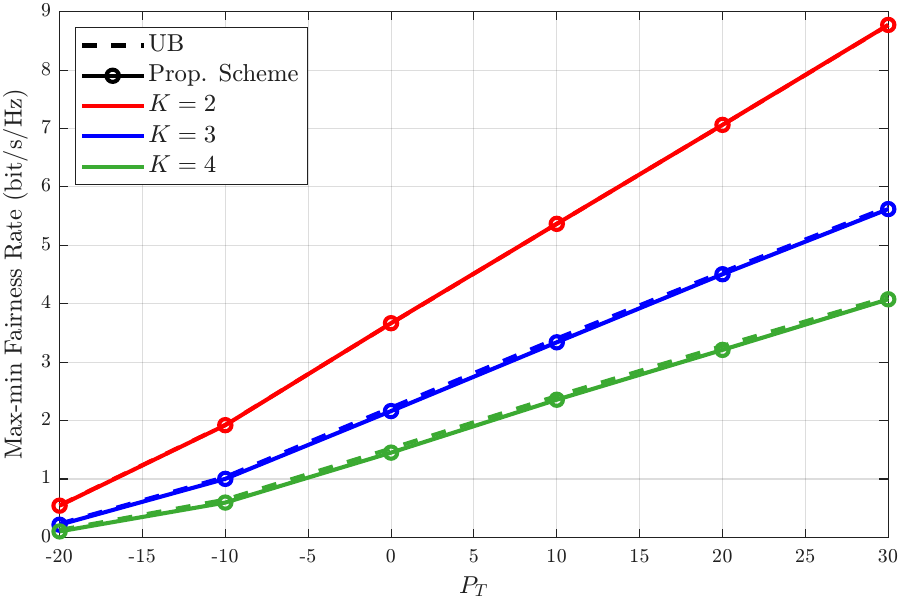}
\vspace{-0.2 in}
\caption{Performance comparison of the proposed scheme and the upper bound versus $P_T$ for different $K$ values.}
\label{fig:PT}
\end{figure}

Figure \ref{fig:PT} considers the effects of transmit power on the max-min rate of \glspl{PAS} serving 2,3 and 4 \glspl{UE} when $N=16$. Since the performance of the proposed algorithm was compared to a range of benchmarks in Fig. \ref{fig:N}, only the proposed method and the perfect combining \gls{UB} are considered in Fig. \ref{fig:PT}.

It can be seen from Fig. \ref{fig:PT} that an increase in transmit power results in a monotonic increase in the max-min rate. As in Fig. \ref{fig:N}, the proposed \gls{JO} approach achieves performance closest to the \gls{UB} when serving a lower number of \glspl{UE} due to the reduced phase alignment requirements. Additionally, the max-min rate grows more steeply with $P_T$ as $K$ decreases. As the transmit power increases, the system becomes interference-limited rather than noise-limited. Systems serving more \glspl{UE} suffer from more inter-user interference and thus less effective \gls{SIC} decoding. As $P_T$ increases, this interference grows alongside the signal power. Therefore, \gls{SINR} of the weakest \gls{UE}, which determines the max-min rate, scales less favorably with  $P_T$ when $K$ is high. This widens the observed gaps between the curves corresponding to different numbers of \glspl{UE}.

\vspace{-0.1 in}
\subsection{The Impact of the Spatial Range of Users}

\begin{figure}
\centering
\includegraphics[width=\columnwidth]{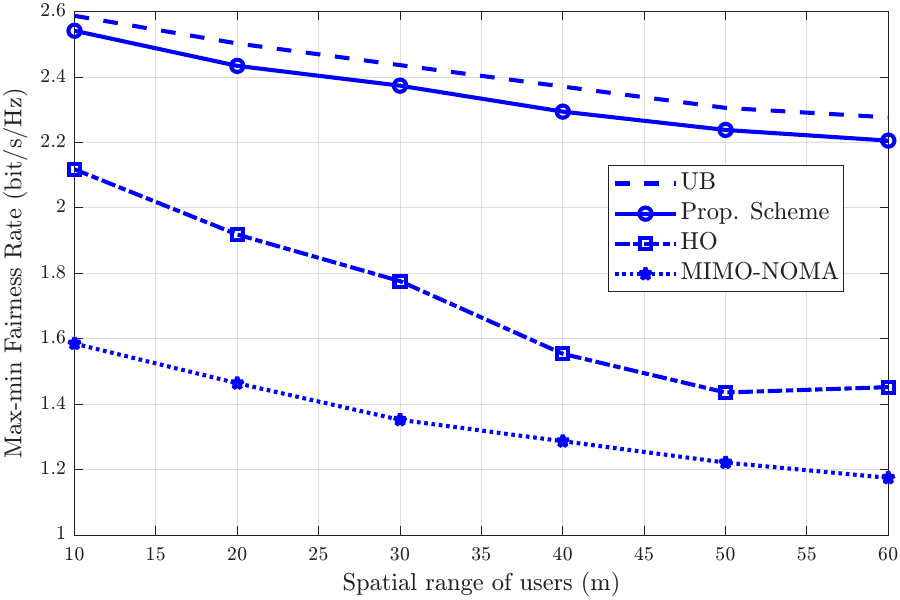}
\vspace{-0.2 in}
\caption{Achievable max-min fairness rate comparison of the proposed scheme and existing benchmarks versus different distance ranges of the users.}
\label{fig:spatial}
\end{figure}

Figure \ref{fig:spatial} considers the effects of \gls{UE} position on the max-min fairness rate. For these results, $N=16$, $P_T=3$ dBm and the positions of $K=3$ \glspl{UE} are varied. It is assumed that \gls{UE} 1 remains at $\mathbf{u}_1=(3,-1,0)$ and the $x$-position of \glspl{UE} 2 and 3 are varied, making their locations $\mathbf{u}_2=(x_2,2,0)$ and $\mathbf{u}_3=(x_3,3,0)$. The $x$ distance between \glspl{UE} 1 and 3 provides the spatial range seen on the $x$-axis, and $x_2=\frac{x_3-3}{2}$. The waveguides are assumed to be $L=x_3$ m long.

It can be seen from Fig. \ref{fig:spatial} that the max-min fairness rate decreases for all schemes when the \glspl{UE} are distributed over a larger area. While \glspl{PA} are able to move to positions near to a \gls{UE}, all users are simultaneously served by all \glspl{PA}. Therefore, as the distance between \glspl{UE} increases, the distance between some of the \glspl{PA} and a given \gls{UE} also increases, resulting in higher path losses in some channels, and a lower max-min fairness rate. However, the performance drop across all methods remains comparatively low, as the received signal at each UE is dominated by contributions from nearby PAs, which are largely unaffected by inter-UE separation.

\vspace{-0.1 in}
\subsection{Computational Time Comparison}

\begin{figure}
\centering
\includegraphics[width=\columnwidth]{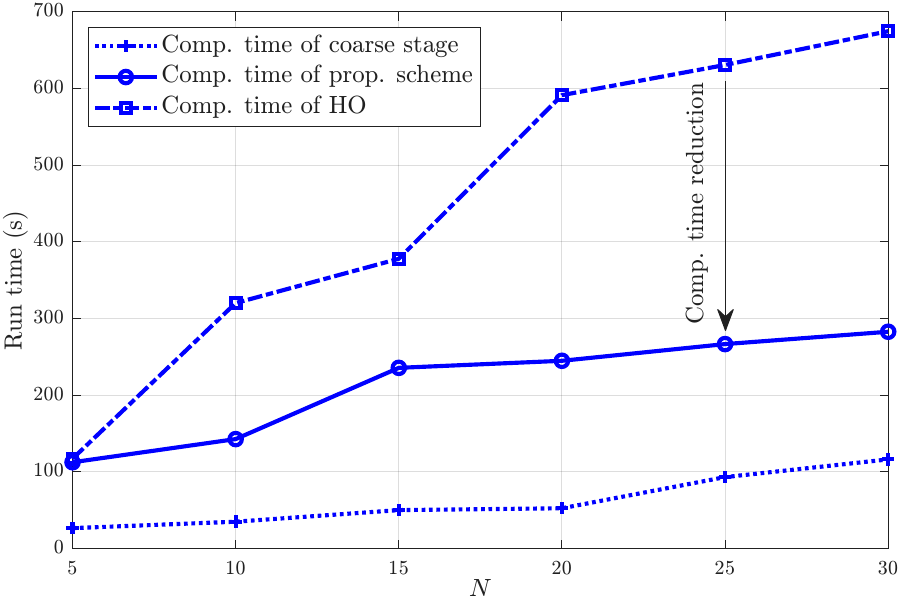}
\vspace{-0.2 in}
\caption{Comparison of Run Time (in seconds) between the Proposed Optimization Scheme and HO Benchmark.}
\label{fig:comp_time}
\end{figure}

Finally, although the number of iterations till convergence in Fig. \ref{fig:conv}(a) looks high, it should be noted that the proposed interior point algorithm for the coarse optimization takes a low convergence time with fast and cheap computations per iteration. Figure \ref{fig:comp_time} shows the computational time of the coarse stage, with $4$ starting points, and the overall time of the proposed algorithm and compares them to the run-time of the \gls{HO} approach with $20$ initial points. The figure shows that the computational run time of the \gls{HO} benchmark is higher than that of the proposed optimization framework, and the gab between both is more pronounced when the number of \glspl{PA}, $N$, increases. Moreover, we shall recall that the achievable performance of our proposed optimization framework is higher than \gls{HO} as discussed in Fig. \ref{fig:N}.

\vspace{-0.05 in}
\section{Conclusions}  \label{Conc}

This paper proposed a max-min rate fairness optimization framework for a \gls{DL} multi-user \gls{PAS}-aided \gls{NOMA} system with multiple \glspl{WG} and multiple \glspl{PA} per \gls{WG}. Since the considered problem is very challenging due to the rapid small-scale fluctuations in the objective function and constraints, a two-stage guided algorithm was developed to jointly optimize the \gls{PA} positions and the transmit precoding matrix. In the first stage, a phase-relaxed coarse optimization procedure, that captures the dominant large-scale geometry of the system, was developed using an \gls{IPA}, exploiting the smooth large-scale behavior of the phase-free channel gains to obtain an effective initialization. In the second stage, the solution was refined through phase alignment and alternating optimization of the \gls{PA} positions and complex transmit precoding to optimize the phase shifts of both \glspl{PA} and complex transmit precoding. The proposed method was shown to outperform the \gls{HO} benchmark while having lower computational time. This is because, unlike population-based metaheuristic approaches, the proposed algorithm leverages the physical structure of the \gls{PAS} channels, avoids repeated random trials and extensive function evaluations, and provides a deterministic and interpretable route to the solution. Simulation results further demonstrated that the proposed \gls{PAS}-aided \gls{NOMA} design achieves substantial gains over a comparable \gls{MIMO} system, mainly due to the proximity gains enabled by flexible \gls{PA} repositioning. The results also showed that the max-min rate improves with both the number of \glspl{PA} and the transmit power, although the improvement becomes less pronounced as the number of \glspl{UE} increases due to stronger inter-user interference.

\section*{Appendix \\ Proof of Lemma~\ref{lem:real_precoding_socp}}
\label{app:proof_real_precoding_lemma}

For real-positive channels, the sign of each precoder can be chosen such that $\mathbf{h}_m^T\mathbf{a}_k\ge 0$ without affecting any received power terms. Hence, the \gls{SINR} constraints in \eqref{eq:real_precoding_power_min_sinr} are equivalently written in the \gls{SOC} form as
\begin{equation}
\mathbf{h}_m^T\mathbf{a}_k
\!\ge\!
\sqrt{\gamma^{\mathrm{Joint}}}
\left\|
\begin{bmatrix}
\mathbf{h}_m^T\mathbf{a}_{k+1}\\
\vdots\\
\mathbf{h}_m^T\mathbf{a}_{K}\\
\sigma_m
\end{bmatrix}
\right\|_2\!\!,
\,\, \forall\, 1\le k\le m\le K,
\label{eq:socp_joint}
\end{equation}
which proves convexity. The convex reformulation of the real precoding power minimization problem can then be given as
\begin{subequations}
\label{eq:real_precoding_soc}
\begin{align}
\min_{\mathbf{A}}\quad & \|\mathbf{A}\|_F^2   \\
\text{s.t.}  \quad  & \eqref{eq:socp_joint}.
\end{align}
\end{subequations}
To characterize the optimal structure, we introduce the quadratic inequalities
\begin{equation}
f_{m,k}(\mathbf{A})\!\triangleq\!
\gamma^{\mathrm{Joint}}\!
\left(\sum_{\ell=k+1}^K\!\!(\mathbf{h}_m^T\mathbf{a}_\ell)^2\!\!+\!\sigma_m^2\!\!\right)
\!-(\mathbf{h}_m^T\mathbf{a}_k)^2
\le 0,
\end{equation}
for all $1\le k\le m\le K$, and associate a nonnegative dual multiplier $\mu_{m,k}$ with each constraint. The Lagrangian becomes
\begin{equation}
\mathcal{L}(\mathbf{A},\boldsymbol{\mu})
=
\sum_{k=1}^K \|\mathbf{a}_k\|_2^2
+\sum_{k=1}^K\sum_{m=k}^K \mu_{m,k} f_{m,k}(\mathbf{A}).
\end{equation}
Using $(\mathbf{h}_m^T\mathbf{a})^2=\mathbf{a}^T\mathbf{h}_m\mathbf{h}_m^T\mathbf{a}$, the Lagrangian separates as
\begin{equation}
\mathcal{L}(\mathbf{A},\boldsymbol{\mu})
=
\sum_{k=1}^K \mathbf{a}_k^T\mathbf{Q}_k(\boldsymbol{\mu})\mathbf{a}_k
+\sum_{k=1}^K\sum_{m=k}^K \mu_{m,k}\gamma^{\mathrm{Joint}}\sigma_m^2,
\end{equation}
where $\mathbf{Q}_k(\boldsymbol{\mu})$ is given in \eqref{eq:Qk_def}. The positive summations in \eqref{eq:Qk_def} represent where $k$ appears as interference in constraints of weaker streams $j<k$, while the negative summations represents where stream $k$ appears as desired signal in constraints $(m,k)$. Define the dual function
\begin{equation}
g(\boldsymbol{\mu})=\inf_{\mathbf{A}}\ \mathcal{L}(\mathbf{A},\boldsymbol{\mu}).
\end{equation}
As $\mathcal{L}$ separates across $\!\{\mathbf{a}_k\}$, the infimum is finite if and only if
\begin{equation}\label{eq:PSD_condition}
\mathbf{Q}_k(\boldsymbol{\mu}) \succeq 0,\qquad \forall t=1,\dots,K,
\end{equation}
in which case the infimum is achieved (e.g., at $\mathbf{a}_k = \mathbf{0}$) and equals
\begin{equation}
g(\boldsymbol{\mu}) = \sum\nolimits_{k=1}^K \sum\nolimits_{m=k}^K \mu_{m,k}\gamma_k \sigma_m^2.
\end{equation}
Otherwise, if any $\mathbf{Q}_k(\boldsymbol{\mu})$ has a negative eigenvalue, then $\inf_{\mathbf{a}_k} \mathbf{a}_k^T \mathbf{Q}_k(\boldsymbol{\mu}) \mathbf{a}_k = -\infty$ and hence $g(\boldsymbol{\mu}) = -\infty$. Therefore, the dual function is finite if and only if $\mathbf{Q}_k(\boldsymbol{\mu})\succeq \mathbf{0}$ for all $k$, yielding a convex dual semidefinite program given in \eqref{eq:dual_SDP}. As the primal \gls{SOCP} in \eqref{eq:real_precoding_soc} is convex and (under feasibility) satisfies Slater's condition, strong duality holds and the dual optimum equals the primal optimum.

The \gls{KKT} stationarity condition with respect to $\mathbf{a}_k$ gives
\begin{equation}
\mathbf{Q}_k(\boldsymbol{\mu}^\star)\mathbf{a}_k^\star=\mathbf{0},\qquad k=1,\dots,K,
\end{equation}
which proves \eqref{eq:ak_nullspace}. In the generic rank-deficient case, the nullspace is one-dimensional, so $\mathbf{a}_k^\star$ can be decomposed as in \eqref{eq:ak_dir_power}. Substituting $\mathbf{a}_k=\sqrt{p_k}\mathbf{v}_k$ into the \gls{SINR} constraints yields a linear program in $\{p_k\}$, whose minimum-sum-power solution is obtained by the backward recursion in \eqref{eq:pK_star_joint}-\eqref{eq:pk_star_joint}. This completes the proof.

\bibliographystyle{IEEEtran}
\bibliography{references}

\end{document}